%% file: main.tex
\documentclass{article}
\usepackage[a4paper, total={6.5in, 10in}]{geometry}
\input{preamble}

\usepackage{authblk}
\title{Steering opinion dynamics through control of social networks}
\author[1,2]{Andrew Nugent}
\author[2,*]{Susana N. Gomes}
\author[2]{Marie-Therese Wolfram}
\affil[1]{MathSys CDT, University of Warwick, Coventry CV4 7AL, UK}
\affil[2]{Mathematics Institute, University of Warwick, Coventry CV4 7AL, UK}
\affil[*]{Corresponding author, susana.gomes@warwick.ac.uk}
\begin{document}

\maketitle

\begin{abstract}
    In this paper we propose a novel control approach for opinion dynamics on evolving networks. The controls modify the strength of connections in the network, rather than influencing opinions directly, with the overall goal of steering the population towards a target opinion. This requires that the social network remains sufficiently connected, the population does not break into separate opinion clusters, and that the target opinion remains accessible. 
    
    We present several approaches to addressing these challenges, considering questions of controllability, instantaneous control and optimal control. Each of these approaches provides a different view on the complex relationship between opinion and network dynamics and raises interesting questions for future research.
\end{abstract}

{\textbf {In this paper we introduce a novel type of control problem for opinion formation on adaptive networks, in which the control variable affects the evolution of the underlying network rather than individuals' opinions. We present various control strategies and analyse under which conditions opinions can/cannot be steered towards a given target, then corroborate and extend our analytical results with computational experiments and a study of optimal controls. We highlight the advantages and disadvantages of each approach, as well as proposing several directions for future research. 
}}
\section{Introduction}
Fundamental models of opinion formation, such as those of DeGroot \cite{degroot1974reaching}, Hegselmann and Krause \cite{hegselmann2002opinion}, and Deffuant \textit{et al.} \cite{deffuant2000mixing} have been repeatedly extended and adapted to create the rich and varied literature of modern opinion dynamics. Many such models extend the idea of bounded confidence to that of a more general, typically non-linear, interaction function \cite{motsch2014heterophilious,piccoli2021control,nugent2023evolving}. The interaction function describes how the distance between individuals' opinions affects whether they interact and the weight this interaction would be given, however it can also be interpreted as a probability of those individuals interacting over some short time period \cite{nugent2023bridging}. This interaction function creates an `instantaneous' network of potential interactions, based entirely on individuals' current opinions, that is sometimes referred to as a state-dependent network \cite{berner2023adaptive}.

Various authors have considered the question of controlling opinion dynamics of this form (sometimes in the more general setting of interacting particle systems) by introducing a control variable that directly affects the evolution of each individual's opinion \cite{albi2014kinetic,qian2011adaptive,albi2016optimal,liu2018polarizability,herty2018suboptimal}. In the context of opinion dynamics this could represent an external effect such as advertising. The typical goal of such a control is to bring the population to consensus, or more specifically to consensus at a given target opinion. As such controls can influence opinions directly they can be highly effective in guiding the population towards a particular target opinion and so the question of optimal control is often considered. Some works have also studied the impact of introducing `strategic agents' whose opinion is controlled \cite{debuse2023automatic}.

To make opinion dynamics more representative of the real world, it is also common to include a social network \cite{lorenz2007continuous,amblard2004role,gabbay2007effects}, which may be static, evolve independently, or evolve coupled to individuals' opinions. It is important to note that the edge-weights of this social network are introduced as additional state variables and so, unlike the interaction function, are not determined by individuals' current opinions. In order to interact, individuals require a non-zero connection in the social network and a non-zero interaction function. Recently, ideas around evolving networks, such as those considered for the Kuraomoto model of coupled oscillators \cite{berner2019multiclusters}, have also been applied to opinion dynamics \cite{nugent2023evolving}. Here the evolution of the network is used to model changing social relationships, which are affected by the history of individuals' interactions. The goal of this paper is to study the potential impact of a control applied to the evolution of the social network. That is, controls gradually alter the extent to which pairs of individuals interact, rather than directly affecting their opinions. Such a control must work within the range of the population's current opinions, while also accounting for the non-linear interaction function, possibility of the population breaking into clusters, and the impact of the initial network structure. 

A related concept of `edge-based' controls, also referred to as a `decentralised adaptive strategy', has previously been addressed with regard to other interacting particle systems \cite{de2009decentralized,delellis2010synchronization,rajapakse2011dynamics}, such as the Kuramoto model \cite{de2008adaptive} and Chua's circuits \cite{yu2012distributed}. A recent review of adaptive dynamical networks, including discussion of these works, can be found in \cite{berner2023adaptive} (note that in this context the term adaptive networks is also used to refer to state-dependent networks, such as those generated by the interaction function). The focus in many prior works has been on providing equations for the evolution of edge weights and showing that these guarantee the stability of the fully synchronised state. This is somewhat different to the setup considered in this paper, in which a control variable will be introduced for each edge weight, and the goal is to determine how these control variables should be set over time to achieve consensus at a particular target. This is closer to the setting considered by Piccoli and Duteil in \cite{piccoli2021control}, in which each individual has a mass representing their influence in the population and control variables are introduced to affect the evolution of these masses. This can be considered as a specific case of network control, in which all edges connecting to a given individual are identical. Here we consider more general network structures and adapt the network dynamics considered in \cite{nugent2023evolving} by replacing the appearance of the interaction function in the network dynamics with these new control variables. 

The remainder of the paper is structured as follows. Section \ref{Section: Model Formulation} describes precisely the mathematical setting and motivates the form of control we will consider. This system is then analysed in Section \ref{Section: Model Analysis} in three ways: Section \ref{Section: Controllability} presents several analytic results about the system's controllability; \ref{Section: Instantaneous control} studies the performance of a candidate control, inspired by the instantaneous control in Piccoli and Duteil \cite{piccoli2021control}; and Section \ref{Section: Optimal control} attempts to improve upon this by considering the question of optimal control. Finally, Section \ref{Section: Conclusion} concludes with possible future research directions. Several proofs and additional numerical examples are provided in Appendices. 

\section{Model Formulation} \label{Section: Model Formulation}

We begin by presenting a general mathematical model for opinion dynamics on an evolving network. We consider a population of $N$ individuals and define $\Lambda = \{1,\dots,N\}$. Fix initial opinions $x(0)=(x_1(0),\dots,x_N(0))\in[-1,1]^N$ and the edge weights of an initial network $w(0)\in[0,1]^{N \times N}$. We assume $w_{ii}(t) = 1$ for all $i\in\Lambda$ and $t\geq0$, meaning that individuals always give their own opinion maximal weight. For clarity we assume individuals are labelled such that $x_1(0)\leq x_2(0) \leq\dots\leq x_N(0)$. Let $x_i(t)$ denote the opinion of individual $i\in\Lambda$ at time $t\geq0$ and $w_{ij}(t)$ the weight of the edge between individuals $i$ and $j$ at time $t$. 

We consider opinion dynamics based on the general formulation in \cite{motsch2014heterophilious,blondel2010continuous,lacker2018mean}. As in \cite{nugent2023evolving} we also introduce dynamics for the edge weights in the form of ODEs. However in this paper, these dynamics are driven by a control variable $u\in L^\infty(\mathds{R}^+ ; \mathds{R}^{N \times N})$. Throughout this paper we consider a control that can affect all edges (except $w_{ii}$) at all times and has perfect information about the current state of the system. We recognise that such a setup is not realistic but it will serve as a starting point for future research. In summary, we consider the following non-linear coupled ODE system
\begin{subequations} \label{eqn: ODE system} 
     \begin{align} 
        \frac{dx_i}{dt} &= \frac{1}{k_i(t)} \sum_{j=1}^N w_{ij}\, \phi(|x_j - x_i|) \,(x_j - x_i) \quad & i\in\Lambda \,\label{eqn: opinion ODE}, \\
        \frac{dw_{ij}}{dt} &= f(u_{ij},w_{ij}) \quad & i,j\in\Lambda \,, i\neq j \,, \label{eqn: weight ODE}
    \end{align}
\end{subequations}
where $f:\mathds{R}\times[0,1]\rightarrow\mathds{R}$ describes the effect of the control and $k_i$ denotes an individuals' in-degree. The in-degree describes the extent to which an individual is influenced by others in the network, and is given by 
\begin{equation} \label{Eqn: Opinion formation ODE}
    k_i(t) = \sum_{j=1}^N w_{ij}(t) \,.
\end{equation}
The function $\phi:[0,2]\rightarrow[0,1]$ in \eqref{eqn: opinion ODE} is an interaction function describing how the difference between individuals' opinions affects the strength/rate of their interactions. 

To ensure that the system \eqref{eqn: ODE system} is well-posed, and that opinions and weights remain in the desired intervals, we introduce the following assumptions on the interaction function and controls.
\begin{assumption} \label{Assumptions on phi}
    The interaction function $\phi$ satisfies 
\begin{enumerate}[label=\alph*)]
    \item $\phi(r) \in [0,1]$ for all $r\in[0,2]$. \label{Assumption: phi positive}
    \item $\phi(0) > 0$. \label{Assumption: phi(0) > 0}
    \item $\phi$ is Lipschitz continuous, with Lipschitz constant $L_\phi$. \label{Assumption: phi integrable}
\end{enumerate}
\end{assumption}

A common interaction function in opinion dynamics is the bounded confidence function
\begin{equation}
    \phi_R(r) = 
    \begin{cases}
        1 & \text{if } 0 \leq r \leq R \,,\\
        0 & \text{if } r > R \,,
    \end{cases}
\end{equation}
for some fixed $R\in[0,2]$. This function has a discontinuity at $R$ and so does not satisfy Assumption \ref{Assumptions on phi}. However, we may consider instead a smoothed version of $\phi_R$, obtained by taking its convolution with a compactly supported mollifier. As discussed in \cite{nugent2023bridging} this corresponds to adding selection noise to the confidence bound, so is a reasonable replacement.  

\begin{definition}
    We call an interaction function $\phi$ a smoothed bounded confidence function with radius $R$, if it satisfies Assumption \ref{Assumptions on phi} and there exists $R>0$ such that $\phi(r) = 1$ for all $r\in[0,R]$. 
\end{definition}

We also introduce assumptions on the form of the control. 

\begin{assumption} \label{Assumptions on f}
    The control function $f:\mathds{R}\times[0,1]\rightarrow\mathds{R}$ satisfies 
\begin{enumerate}[label=\alph*)]
    \item $f(0,w_{ij}) = 0$ for all $w_{ij}\in[0,1]$, so that edge weights remain constant if uncontrolled. \label{Assumption: no effect for no control}
    \item $f(u_{ij},0) \geq 0$ for all $u_{ij}\in\mathbb{R}$, so that edge weights remain non-negative. \label{Assumption: f positive at w=0}
    \item $f(u_{ij},1) \leq 0$ for all $u_{ij}\in\mathbb{R}$, so that edge weights do not exceed $1$. \label{Assumption: f negative at w=1}
    \item $f$ is bounded and integrable. \label{Assumption: f bounded}
\end{enumerate}
\end{assumption}

This final assumption on the boundedness of $f$ highlights a key feature of this problem: edge weights $w_{ij}$ change continuously in time with a finite speed, meaning edges cannot be switched on/off instantaneously. As a result, the maximal value of $f$ will play a major role in determining the controllability of \eqref{eqn: ODE system}. Note this does not mean that controls cannot be switched on/off instantaneously, but rather that their effect is not instantaneous. 

When Assumption \ref{Assumptions on phi} and Assumption \ref{Assumptions on f} hold, Proposition 3.1 from \cite{nugent2023evolving} ensures that $x_i(t)\in[-1,1]$ and $w_{ij}(t)\in[0,1]$ for all $i,j\in\Lambda$ and $t\geq0$.

This work focuses on a form of control function $f$ motivated by the memory weight dynamics discussed in \cite{nugent2023evolving} and similar to those discussed in \cite{gkogkas2021continuum}, 
\begin{align} \label{Eqn: Memory weight controls}
    f(u_{ij},w_{ij}) &= s(u_{ij}) \, \big( \ell(u_{ij}) - w_{ij} \big) \,.
\end{align}
Here $s:\mathds{R}\rightarrow\mathds{R}^+$ describes the rate at which $w_{ij}$ changes when controlled and $\ell:\mathds{R}\rightarrow[0,1]$ describes the target towards which $w_{ij}$ is directed. Similarly to \cite{piccoli2021control}, in which an opinion dynamics model with evolving masses is considered, we aim to control the population to consensus at a desired value $x^*\in[-1,1]$. The method for selecting the control $u\in L^\infty(\mathds{R}^+ ; \mathds{R}^{N \times N})$ in order to achieve this target is referred to as a control strategy. 

\begin{remark}
    Other forms of control function could be considered, and the controllability of \eqref{eqn: ODE system} is naturally dependent on this choice. As we will see throughout this paper, controlling \eqref{eqn: ODE system} poses several major challenges, hence we study a control function that has the potential to significantly alter the network's structure. By contrast, in \cite{nugent2023evolving} the authors also introduce logistic weight dynamics, motivating a control function of the form 
    \begin{align} \label{Eqn: Logistic weight controls}
        f(u_{ij},w_{ij}) &=  u_{ij} \,w_{ij} \, (1 - w_{ij}) \,,
    \end{align}
    in which $u_{ij}$ controls the rate at which $w_{ij}$ is increasing or decreasing. This form of control cannot add new edges, creating a strong dependence on the initial network that severely limits both the control's effectiveness and our ability to analyse it. 
\end{remark}

To motivate our assumptions on $\phi$, $x^*$ and $x(0)$, we first establish some basic conditions for consensus. 

\begin{definition} \label{Definition: consensus}
    We say the population reaches consensus if, for all $i,j\in\Lambda$,
    \begin{equation*}
        \lim_{t\rightarrow\infty} x_i(t) = \lim_{t\rightarrow\infty} x_j(t) \,.
    \end{equation*}
    Moreover, we say the population reaches consensus at a point $x^*\in[-1,1]$ if, for all $i\in\Lambda$,
    \begin{equation*}
        \lim_{t\rightarrow\infty} x_i(t) = x^*.
    \end{equation*}
\end{definition}

This definition is commonly used when considering consensus or synchronisation \cite{delellis2010synchronization}. Denote the minimum and maximum opinions in the population by 
\begin{align*}
    x_m(t) = \min_{i\in\Lambda} x_i(t) \,,\quad
    x^M(t) = \max_{i\in\Lambda} x_i(t) \,,
\end{align*}
and the opinion diameter $D(t) = x^M(t) - x_m(t)$. Definition \ref{Definition: consensus} of consensus is then equivalent to requiring $D(t)\rightarrow0$ as $t\rightarrow\infty$. The following provides a useful characterisation of consensus at a point. 

\begin{proposition} \label{Propostion: x^* must remain in the opinion interval}
    Assume that the population reaches consensus. Then the population reaches consensus at a point $x^*\in[-1,1]$ iff $x^*\in [x_m(t),x^M(t)]$ for all $t\geq 0$. 
\end{proposition}
\begin{proof}
    $(\Leftarrow)$ Assume $x^*\in [x_m(t),x^M(t)]$ for all $t\geq 0$. As the population reaches consensus, $\lim_{t\rightarrow\infty} x_m(t) = \lim_{t\rightarrow\infty} x^M(t) \,,$
    so $D(t)\rightarrow0$. As $x^*\in [x_m(t),x^M(t)]$, $|x_m(t) - x^*| < D(t)$ so $x_m(t)\rightarrow x^*$. Therefore $x_i(t)\rightarrow x^*$ for all $i\in\Lambda$ and we have consensus at $x^*$. 

    $(\Rightarrow)$ Assume that the population reaches consensus at $x^*$, but that there exists a time $s\geq0$ at which $x^*\notin [x_m(s),x^M(s)]$. As $x_m$ is increasing and $x^M$ is decreasing (see Proposition 2.1 in \cite{nugent2023evolving}), we have that $x^*\notin [x_m(t),x^M(t)]$ for all $s \geq t$. If $x^* < x_m(s)$ then $|x_1(t) - x^*| \geq |x_m(s) - x^*| > 0$ for all $t\geq0$. Similarly if $x^* > x^M(s)$ then $|x_1(t) - x^*| \geq |x^M(s) - x^*| > 0$ for all $t\geq0$. In both cases this makes convergence of $x_1$ to $x^*$ impossible, giving a contradiction. 
\end{proof}

From Proposition \ref{Propostion: x^* must remain in the opinion interval} it is clear that we will at least require $x^*\in[x_m(0),x^M(0)]$ to have any hope of reaching consensus at $x^*$. Of course, we will also require that consensus itself is possible. We introduce the following definitions to help clarify when this is not the case.  

\begin{definition}
     For a given interaction function $\phi:[0,2]\rightarrow[0,1]$ we denote the set of roots of $\phi$ by 
    \begin{equation}
        \mathcal{R}_\phi = \{ r \in [0,2] : \phi(r) = 0 \}.
    \end{equation}
    If $\mathcal{R}_\phi$ is empty then define $r^* = 2$, otherwise let $r^* = \inf (\mathcal{R}_\phi)$. 
\end{definition}

By Assumption \ref{Assumptions on phi}, $r^* > 0$. If $\phi$ is a smoothed bounded confidence function with radius $R$, we will also have that $0<R<r^*$.

\begin{definition}
     For a given $r>0$, an opinion profile $x(t)$ is called an $r$-chain if $|x_{i+1}(t) - x_i(t)|<r$ for all $i\in\Lambda$. 
\end{definition}

\begin{proposition} \label{Proposition: gaps of size r^* cannot be closed}
    If $\phi$ is decreasing and there exists $i\in\Lambda$ such that $|x_{i+1}(0) - x_i(0)|>r^*$, then the population will not reach consensus. 
\end{proposition}
\begin{proof}
    See, for example, Proposition 3.2 in \cite{nugent2023evolving}. The fundamental idea is that once a gap in the opinion profile of size bigger than $r^*$ appears, the closest individuals on either side of this gap will be unable to move closer than $r^*$ as they will encounter a root of $\phi$. 
\end{proof}

Motivated by Proposition \ref{Propostion: x^* must remain in the opinion interval} and Proposition \ref{Proposition: gaps of size r^* cannot be closed}, we introduce the following assumptions. 
\begin{assumption} \label{Assumptions on x(0) and x^*} 
    The initial opinions $x(0)$ and consensus target $x^*$ satisfy
    \begin{enumerate}[label=\alph*)]
        \item $x^*\in[x_m(0),x^M(0)]$. \label{Assumption: x^* in initial opinion interval}
        \item $x(0)$ is an $r^*$-chain. \label{Assumption: x(0) is an r^* chain} 
    \end{enumerate}
\end{assumption}

These Assumptions ensure we are operating in an environment in which the opinions and networks are well-defined and control to consensus may be feasible. The question of which consensus targets are achievable then depends on the initial network and the speed at which controls can alter the network structure.

\section{Model Analysis} \label{Section: Model Analysis}

In this Section we consider controls of the form \eqref{Eqn: Memory weight controls} for given functions $s:\mathds{R}\rightarrow\mathds{R}^+$ and $\ell:\mathds{R}\rightarrow[0,1]$. We assume that $s(0)=0$, so that $f$ satisfies Assumption \ref{Assumptions on f}. 

We first show which consensus targets are guaranteed to be achievable when using an edge-creating control $u^+\in\mathds{R}^+$ and an edge removing control $u^-\in\mathds{R}^-$, then investigate the performance of a candidate instantaneous control, and finally address questions of optimal control for example $s$ and $\ell$ functions. 

\subsection{Controllability} \label{Section: Controllability}

The first result of this Section, Proposition \ref{Proposition: Uncontrollability}, shows that the model setup and Assumptions described in Section \ref{Section: Model Formulation} do not guarantee controllability, indeed we can always find situations in which the system is not controllable.

\begin{proposition} \label{Proposition: Uncontrollability}
    Fix some $N > 1$. Let the interaction function $\phi$ satisfy Assumption \ref{Assumptions on phi} and the control function $f(u_{ij},w_{ij})$ satisfy Assumption \ref{Assumptions on f}. Then, for any $x^*\in(-1,1)$ there exist initial opinions $x(0)\in[-1,1]^N$ and initial edge weights $w(0)\in[0,1]^{N\times N}$ satisfying Assumption \ref{Assumptions on x(0) and x^*} for which control to consensus at $x^*$ is not possible. 
\end{proposition}
\begin{proof}
    Without loss of generality, assume $x^*\in(-1,0]$. We will construct a range of $x(0)$ and $w(0)$ values for which control is not possible. Take $x_1(0)\in[x^* - \varepsilon, x^*)$ for some small $\varepsilon$ to be determined. Assume that $\varepsilon$ is sufficiently small that $x_1(0)\in[-1,1]$. Let $\nu = \min(1 - x^*, \frac{1}{2}r^*) > 0$ and take $x_i(0)\in[x^* + \frac{1}{2}\nu, x^* + \nu]$ for all $i\in\Lambda\setminus\{1\}$. This gives a setup in which $x_1$ is `far away' from the rest of the population, while maintaining the required $r^*$-chain. Take $w_{1j}(0) \geq \frac{1}{2}$ for all $j\in\Lambda$. The other entries of $w(0)$ may take any value in $[0,1]$. These initial conditions and $x^*$ satisfy Assumption \ref{Assumptions on x(0) and x^*}. 
    
    We now show that if $\varepsilon$ is chosen to be sufficiently small, there exists a time $\tau>0$ such that $x_i(\tau) > x^*$ for all $i\in\Lambda$, hence by Proposition \ref{Propostion: x^* must remain in the opinion interval} consensus at $x^*$ is impossible. This is achieved using bounds on $w_{1j}$, $\phi(|x_j - x_i|)$ and $x_i$.

    As $f$ satisfies Assumption \ref{Assumptions on f} it is bounded, hence there exists a constant $\mathcal{S} > 0$ such that $|f(u_{ij},w_{ij})| < \mathcal{S}$ for all $u_{ij}\in\mathds{R}$ and $w_{ij}\in[0,1]$. Hence 
    \begin{equation*}
        |w_{1j}(t) - w_{1j}(0)| \leq \int_0^t \big| f\big(u_{ij}(s),w_{ij}(s)\big) \big| ds \leq \mathcal{S}t \,.
    \end{equation*}
    As $w_{1j}(0)\geq \frac{1}{2}$, for $t\leq \frac{1}{4\mathcal{S}}$ we have $w_{1j}(t) \geq \frac{1}{4}$ for all $j\in\Lambda$.

    As $\varepsilon$ can be chosen sufficiently small that $D(0)<r^*$, there exists a constant $c > 0$ such that $\phi(r) > c$ for all $r\in[0,D(0)]$. As $D(t) \leq D(0)$ for all $t\geq0$, $\phi\big(|x_j(t) - x_i(t)|\big) > c$ for all $i,j\in\Lambda$ and $t\geq 0$. 

    Finally, for any $i\in\Lambda$, consider
    \begin{align*}
        \bigg| \frac{dx_i}{dt} \bigg| 
        &= \bigg| \frac{1}{k_i} \sum_{j=1}^N w_{ij}\,\phi\big(|x_j(t) - x_i(t)|\big)\,(x_j - x_i) \bigg|  \\
        &\leq \frac{1}{k_i} \sum_{j=1}^N w_{ij}\,\phi\big(|x_j(t) - x_i(t)|\big)\,|x_j - x_i| \\
        &\leq \sum_{j=1}^N |x_j - x_i| \\
        &\leq N r^* \,.
    \end{align*}
    Hence $x_i(t) \geq x_i(0) - t N r^*$. Specifically, for $i\in\Lambda\setminus\{0\}$ and $t\leq\frac{\nu}{4Nr^*}$, $x_i(t) \geq x^* + \frac{1}{4}\nu$. 

    Combining these bounds we have the following 
    \begin{align*}
        \frac{dx_1}{dt} &= \frac{1}{k_1} \sum_{j=1}^N w_{1j}\,\phi\big(|x_j(t) - x_1(t)|\big)\,(x_j - x_1) \\
        &\geq \frac{1}{N} \sum_{j=1}^N \frac{1}{4}\,c\,(x_j - x_1) & \text{for } t < \frac{1}{4\mathcal{S}} \\
        &\geq \frac{1}{N} \sum_{j=1}^N \frac{1}{4}\,c\,\bigg(x^* + \frac{1}{4}\nu - x^*\bigg) & \text{for } t < \min\bigg(\frac{1}{4\mathcal{S}}, \frac{\nu}{4Nr^*}\bigg) \\
        &= \bigg(\frac{c\nu}{16} \bigg) \bigg(\frac{N-1}{N}\bigg) & \text{for } t < \min\bigg(\frac{1}{4\mathcal{S}}, \frac{\nu}{4Nr^*}\bigg) \,.
    \end{align*}
    Thus for $t \leq \tau := \min\big(\frac{1}{4\mathcal{S}}, \frac{\nu}{4Nr^*}\big) $ we have $x_1(t) \geq x_1(0) + \big(\frac{c\nu}{16} \big) \big(\frac{N-1}{N}\big) t$. Note that the definition of $\tau$ is independent of $\varepsilon$, so by taking $\varepsilon < \big(\frac{c\nu}{16} \big) \big(\frac{N-1}{N}\big) \tau$ we ensure that $x_1(\tau) > x^*$. In addition, $x_i(\tau) \geq x^* + \frac{1}{4}\nu > x^*$ for $i\in\Lambda\setminus\{0\}$. Hence at time $\tau$ all opinions lie above $x^*$, so control to consensus at $x^*$ is impossible. 
\end{proof}

Hence it is not possible to provide a control function $f$ satisfying Assumption \ref{Assumptions on f}, and so not possible to provide a control strategy to determine $u$, that guarantees controllability for all $x^*$. Instead we consider the initial conditions $x(0)$ and $w(0)$, as well as the consensus target $x^*$, to be fixed and ask when the system can be controlled for these fixed values. 

We begin by considering the simple case in which $w_{ij}(0)=0$ for all $i\neq j$, and the control acts only to create new edges. Recall that $w_{ii}$ is always equal to $1$ for all $i\in\Lambda$.

\begin{definition}
    A network $w\in[0,1]^{N\times N}$ is called empty if $w_{ij}=0$ for all $i\neq j$ and non-empty if there exists $i\neq j$ such that $w_{ij} > 0$. 
\end{definition}

\begin{proposition} \label{Proposition: Controllability from empty intial network}
    Let $\phi$ satisfy Assumption \ref{Assumptions on phi}, $x(0)$ and $x^*$ satisfy Assumption \ref{Assumptions on x(0) and x^*}, and $w(0)$ be an empty network. Assume also that there exists a control value $u^+\in\mathbb{R}$ for which  
    \begin{align} \label{Eqn: u^+ definition}
        s^+ = s(u^+) > 0  \,,\quad
        \ell^+ = \ell(u^+) > 0 \,.
    \end{align}
    That is, the control $u^+$ can be used to create new edges. Then there exists a control $u\in L^\infty(\mathds{R}^+; \{0,u^+\}^{N\times N})$ such that the solution of \eqref{eqn: ODE system} reaches consensus at $x^*$. 
\end{proposition}
\begin{proof}
    The approach of the proof is to construct such a control. We proceed in three steps, firstly addressing the simplest case of $N=2$ individuals. In the second step we reduce the general case to the $N=2$ case by identifying the the closest individuals to $x^*$ above and below and gathering the rest of the population towards these two individuals. Finally in the third step we apply the $N=2$ case once this gathering is complete. 

    \textbf{Step 1:} Consider first the case that $N = 2$. Clearly, if $x_1(0) = x_2(0)$ then we immediately have consensus at $x^* = x_1(0) = x_2(0)$, so assume $x_1(0) < x_2(0) $. We then define times $T_{12}$ and $T_{21}$ and set the controls $u_{ij}$ for $ij = 12,\,21$ according to
    \begin{equation} \label{Eqn: Controllability from empty network control scheme}
        u_{ij}(t) = 
        \begin{cases}
            u^+ & \text{if } 0 \leq t \leq T_{ij} \,,\\
            0 & \text{if } t > T_{ij} \,.
        \end{cases}
    \end{equation}
    That is, $T_{ij}$ gives the time at which $u_{ij}$ is switched off, hence a time after which $w_{ij}$ remains fixed.
    We also define
    \begin{equation} \label{Eqn: F definition}
        F(T_{12},T_{21}) = \lim_{t\rightarrow\infty} x_1(t) \,,
    \end{equation}
    for the controlled dynamics with \eqref{Eqn: Controllability from empty network control scheme}. As long as $T_{12}$ and $T_{21}$ are not both zero there will exist an edge between $x_1$ and $x_2$, in addition $|x_1(0) - x_2(0)| < r^*$, so for $(T_{12},T_{21})\in\mathds{R}^2\setminus(0,0)$ the system reaches consensus. $F(T_{12},T_{21})$ therefore gives the consensus opinion. We now show that there exists values of $T_{12}$ and $T_{21}$ such that $F(T_{12},T_{21})=x^*$, meaning the system reaches consensus at the desired target $x^*$. 
    
    First consider $x^* \leq \frac{1}{2}(x_1(0)+x_2(0))$ and fix a value of $T_{21} > 0$. By Lemma \ref{Lemma: F continuity} in Appendix \ref{Appendix: Proofs}, for a fixed $T_{21}>0$ the function $F(\cdot,T_{21})$ is continuous. As $w_{12}(0) = w_{21}(0) = 0$, $F(0,T_{21}) = x_1(0)$ and $F(T_{21},T_{21}) =\frac{1}{2}(x_1(0)+x_2(0))$. Hence by the intermediate value theorem, for any given $x^*\in[x_1(0),\frac{1}{2}(x_1(0)+x_2(0))]$ there exists $(T_{12},T_{21})$ such that the system reaches consensus at $x^*$. By an analogous argument, the same can be achieved for $x^*\in[\frac{1}{2}(x_1(0)+x_2(0)), x_2(0)]$. This proves the claim in the case $N = 2$. 
    
    \textbf{Step 2:} For $N > 2$ we begin by gathering individuals towards points on either side of $x^*$. Assume there are no individuals $i$ with $x_i(0) = x^*$. Define individual $a$ such that $x_a(0) < x^*$ and $x_i(0) \leq x_a(0)$ for all individuals $i$ with $x_i(0) < x^*$. That is, $a$ is the closest individual whose initial opinion is strictly below $x^*$. If this initial opinion is not unique then choose the individual with the highest index. Similarly define individual $b$ to be the closest individual whose initial opinion is strictly above $x^*$. If this initial opinion is not unique then choose the individual with the lowest index. As there are no individuals with $x_i(0) = x^*$, we will have $b = a+1$ and so $|x_a(0) - x_b(0)| < r^*$. All individuals with $x_i(0) \leq x_a(0)$ will be brought upwards towards $x_a(0)$, while all individuals with $x_i(0) \geq x_b(0)$ will be brought down towards $x_b(0)$. This is visualised in the diagram in Figure \ref{fig:step 2 diagram}. We will then apply the $N=2$ case to $x_a$ and $x_b$. 

    \input{Tikz/tikz_step2}
    
    Recall that opinions are ordered so that $x_1(0)\leq x_2(0)\leq\dots\leq x_N(0)$. To explain how individuals are gathered upwards towards $x_a$ we assume $x_1(0) < x_a(0)$. If this is not the case then there must be individuals with $x_i(0) \geq x_b(0)$, where a similar argument holds. We know $|x_1(0) - x_2(0)|<r^*$, so setting $u_{12} = u^+$ will cause $x_1$ to move towards $x_2$ as 
    \begin{equation*}
        \frac{dx_1}{dt} = \frac{1}{k_1} \sum_{j=1}^N w_{1j}\, \phi(x_j - x_1) \,(x_j - x_1) = \frac{w_{21}}{1 + w_{21}} \phi(x_2 - x_1) \,(x_2 - x_1) > 0 \,,
    \end{equation*}
    for $t > 0$. As no controls have yet been applied to $x_2$, $x_2(t) = x_2(0)$. Hence for any $\varepsilon > 0$ we will eventually have $|x_1(t) - x_2(t) | < \varepsilon$. Specifically we can choose $\varepsilon$ sufficiently small that $|x_1(t) - x_3(0) | < r^*$, as we know $|x_2(t) - x_3(t) | = |x_2(0) - x_3(0) | < r^*$. Once this has occurred we set $u_{23} = u^+$ and repeat (each time waiting until $|x_1(t) - x_j(0) | < r^*$). In this way, all individuals $i$ with $x_i(0) \leq x_a(0)$ can be sequentially gathered upwards towards $x_a(0)$ without breaking the $r^*$-chain. In a similar way, all individuals $j$ with $x_j > x_b(0)$ can be gathered downwards towards $x_b(0)$. 
    
    \textbf{Step 3:} Once this gathering process is complete (at a time denoted by $T$) we will have $D(T) < r^*$ and so $|x_i(t) - x_j(t)| < r^*$ for all $i,j$ and all $t \geq T$. This can be seen in the diagram in Figure \ref{fig:step 3 diagram} and also in the example in Figure \ref{fig: Controllability from empty network example} in which $T$ is approximately 10. Using the $N=2$ case, $x_a$ and $x_b$ can be controlled to consensus at $x^*$ and thus, by the chain of connections, all individuals will be brought to consensus at $x^*$.  

    \input{Tikz/tikz_step3}

    If there is an individual $i$ with $x_i(0) = x^*$ then instead gather all individuals towards $x_i$, while leaving individual $i$ at their initial position (this is the same as setting $a=b=i$). As the network is initially empty, individual $i$ will not move if no controls are applied to their weights. Here there is no need to apply the $N=2$ case. 
    
\end{proof}

Figure \ref{fig: Controllability from empty network example} shows a numerical example of the control described in Proposition \ref{Proposition: Controllability from empty intial network} applied with $s^+ = \ell^+ = 1$. We use a smoothed bounded confidence interaction function with $r^* = 0.6$ and $R = 0.3$, which would typically lead to opinion clustering if no control was applied. A population of size $N = 50$ is used with each $x_i(0)$ chosen uniformly at random in the interval $[-1,1]$. Individuals $a$ and $b$ are identified and the problem is first solved for these $N=2$ individuals. This is done by repeatedly testing values of $T_{12}$ and $T_{21}$ until a suitable pair is found (as the ODE is solved numerically with a fixed timestep over a finite time interval the exact optimal values cannot be used). The gathering approach described is then implemented until the opinion diameter is sufficiently small, then the $N=2$ case is used to bring the population to consensus. 
\begin{figure}[ht!]
    \centering
    \includegraphics[width = .9\linewidth, trim = {1cm 0cm 1cm 1cm},clip]{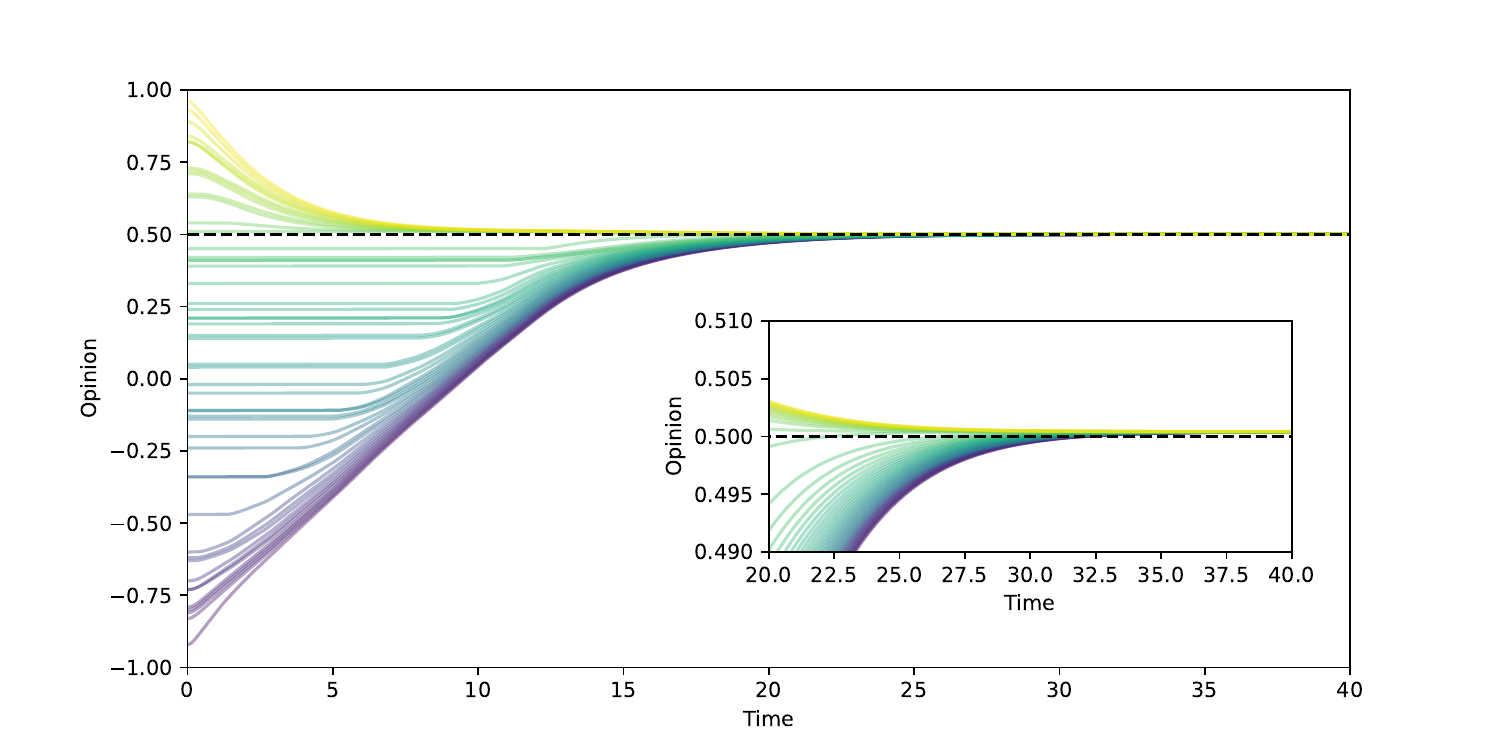}
    \caption{Example of the control method described in Proposition \ref{Proposition: Controllability from empty intial network}. Beginning with an empty network, edges are created to gather the population closer to individuals near the target opinion $x^* = 0.5$. The solution for the $N=2$ case is then used to ensure that those closest to $x^*$ reach consensus at exactly this point. Opinion trajectories are coloured according to individuals' initial opinions. The target opinion $x^*$ is indicated by a black dashed line. The inset plot shows the final approach to consensus at $x^*$.}
    \label{fig: Controllability from empty network example}
\end{figure}

When considering non-empty initial networks, Proposition \ref{Propostion: x^* must remain in the opinion interval} and Proposition \ref{Proposition: gaps of size r^* cannot be closed} show that it is vital that the control can act sufficiently quickly to prevent $x^*$ leaving the interval $[x_m(t),x^M(t)]$ or gaps of size $r^*$ appearing in the opinion profile. The following lemma provides a useful bound on the distance an individual can travel if their edge weights are reduced exponentially quickly after some time $t_0$.  

\begin{lemma} \label{Lemma: bounds on distance travelled}
    Let $\phi$ satisfy Assumption \ref{Assumptions on phi}. For a starting time $t_0\geq 0$, fix opinions $x(t_0)\in[-1,1]^N$, and network $w(t_0)\in[0,1]^{N \times N}$. Assume there exists a control value $u^-\in\mathbb{R}$ for which  
    \begin{align} \label{Eqn: u^- definition}
        \mathcal{S} = s(u^-) > 0 \,,\quad
        \ell(u^-) = 0 \,.
    \end{align}
    Then, for an individual $i$, setting $u_{ij}(t) = u^-$ for all $j\neq i$ and $t\geq t_0$ gives
    \begin{equation*}
        |x_i(t) - x_i(t_0)| < \frac{D(t_0)\, (k_i(t_0)-1)}{\mathcal{S}} \,.
    \end{equation*}
\end{lemma}
\begin{proof}
    For simplicity take $t_0 = 0$. Set $u_{ij}(t) = u^-$ for all $j\neq i$ and $t\geq0$. Then the solution to \eqref{eqn: weight ODE} is 
    \begin{equation}
        w_{ij}(t) = w_{ij}(0)\, e^{-\mathcal{S}t} \,,
    \end{equation}
    and so for any $t\geq 0$
    \begin{align*}
        \big| x_i(t) - x_i(0) \big| 
        &\leq \sum_{j\neq i} \int_{0}^t  \frac{1}{k_i(s)} \,  w_{ij}(s) \, \phi\big(|x_j(s) - x_i(s)|\big) \,\big|x_j(s) - x_i(s)\big| \,ds\, \\
        &\leq \sum_{j\neq i} D(0) \int_{0}^t   w_{ij}(s) \,ds\, \\
        &\leq D(0)\,\frac{1}{\mathcal{S}} \bigg(1 - e^{-\mathcal{S}t}\bigg) \sum_{j\neq i} w_{ij}(0) \\
        &\leq \frac{D(0)\, (k_i(0)-1)}{\mathcal{S}} \,.
    \end{align*} 
\end{proof}

Lemma \ref{Lemma: bounds on distance travelled} shows that if $\mathcal{S}$ is sufficiently large, an individual's opinion can be trapped within a small interval around its current position. This can be used to prevent individuals breaking an $r$-chain or bound individuals above or below $x^*$, and so will prove crucial in showing controllability. 

\begin{theorem} \label{Theorem: Controllability from non-empty intial network}
    Let $\phi$ be a smoothed bounded confidence function with radius $R$. Fix initial conditions $x(0)\in[-1,1]^N$, $w(0)\in[0,1]^{N \times N}$ and a consensus target $x^*$, satisfying Assumption \ref{Assumptions on x(0) and x^*}. Also assume that $x^*\neq x_m(0),\,x^M(0)$. Assume $u^+$ and $u^-$ exist as defined in \eqref{Eqn: u^+ definition} and \eqref{Eqn: u^- definition} respectively. Then for $\mathcal{S}$ sufficiently large, there exists a control $u\in L^\infty(\mathds{R}^+; \{0,u^+,u^-\}^{N\times N})$ such that the solution of \eqref{eqn: ODE system} reaches consensus at $x^*$. Specifically, define $d_1$ and $d_2$ by 
    \begin{align} \label{Eqn: d_1 and d_2 vals}
        d_1 = \max_{i\in\{1,\dots,N-1\}} \frac{k_i(0) + k_{i+1}(0) - 2}{r^* - |x_i(0) - x_{i+1}(0)|}\,,\quad
        d_2 = \max_{\substack{i\in\Lambda \\ x_i(0)\neq x^* }}  \frac{k_i(0)-1}{|x_i(0) - x^*|} \,,
    \end{align}
    in which case we take 
    \begin{equation} \label{Eqn: statement S bound}
        \mathcal{S} > D(0) \max\bigg\{ d_1, d_2, \frac{4 (N-1)}{R} \bigg\} \,.
    \end{equation}
\end{theorem}
\begin{proof}
    The approach of this proof is similar to Proposition \ref{Proposition: Controllability from empty intial network}. Again we proceed in several steps. In Step 1 we identify individuals $a$ and $b$ whose initial opinions are the closest below/above $x^*$ and gather all other individuals towards $x_a$ or $x_b$. This gives an opinion radius below $r^*$, which is further reduced below $R$ by temporarily connecting individuals $a$ and $b$. This step requires two cases, as the setup is slightly different if an individual has an initial opinion of exactly $x^*$. In Step 2, control to consensus at $x^*$ is managed by specifying controls for the central individuals $a$ and $b$. A key difference from Proposition \ref{Proposition: Controllability from empty intial network} is that any non-zero initial edges cannot be removed in finite time, so Lemma \ref{Lemma: bounds on distance travelled} must be applied to bound this potential error. The applications of Lemma \ref{Lemma: bounds on distance travelled} provide the lower bound \eqref{Eqn: statement S bound} on $\mathcal{S}$.

    \textbf{Step 1:} Firstly note that, as $x(0)$ is an $r^*$-chain, $d_1$ is well-defined. As $\mathcal{S} > \Big( D(0)\,\max\big\{ d_1, d_2 \big\} \Big)$ by Lemma \ref{Lemma: bounds on distance travelled}, setting $u_{ij} = u^-$ for all $t\geq0$ gives
    \begin{equation} \label{Eqn: pairwise distance ineq}
        |x_i(t) - x_{i+1}(t)| < r^*
    \end{equation}
    for all $i\in\Lambda$ and 
    \begin{equation} \label{Eqn: non-crossing x^* ineq}
        |x_i(t) - x^*| > 0
    \end{equation}
    for all $i\in\Lambda$ with $x_i(0)\neq x^*$. Note that if $r^*=2$ then \eqref{Eqn: pairwise distance ineq} is always satisfied and the condition on $d_1$ can be removed from \eqref{Eqn: statement S bound}.

    We now describe the process by which individuals are gathered towards a pair of individuals near $x^*$.

    \textbf{Case 1:} $x_i(0)\neq x^*$ for all $i\in\Lambda$.
    
    Here we can identify individuals $a$ and $b$ as in Proposition \ref{Proposition: Controllability from empty intial network}. As before we use $u^+$ to sequentially gather more extreme individuals towards $x_a$ and $x_b$. Here $x_a(t)$ and $x_b(t)$ will not be fixed but, by setting all other controls to $u^-$, \eqref{Eqn: pairwise distance ineq} and \eqref{Eqn: non-crossing x^* ineq} ensure they do not cross $x^*$ or move in a way that could break the $r^*$-chain. The gathering process described in Proposition \ref{Proposition: Controllability from empty intial network}, combined with the exponentially fast removal of all other edges, also ensures that no other pair of individuals breaks the $r^*$-chain. We say that this gathering process is complete when all individuals are within a distance $\varepsilon > 0$ of their target (either $x_a$ or $x_b$). As $x_i(0)\neq x^*$ for all $i\in\Lambda$, we will have $b=a+1$ and so $|x_a(t)-x_b(t)|<r^*$ for all $t\geq0$. Hence, by choosing $\varepsilon$ sufficiently small, there exists a time $T$ at which $x_a(T) < x^*$, $x_b(T) > x^*$ and $D(T)<r^*$. 
    
    From here we can simplify the problem by reducing the diameter from $r^*$ to below $R$, meaning that the interaction function $\phi$ becomes equal to $1$. If $|x_a(T) - x^*| < R/2$ and $|x_b(T) - x^*| < R/2$ then $|x_a(T) - x_b(T)| < R$ and we are done. If only $|x_a(T) - x^*| < R/2$, then create an edge from $b$ to $a$, bringing $x_a$ and $x_b$ within a distance $R$. This can be achieved while maintaining $|x_b - x^*| > R/2$. If only $|x_b(T) - x^*| < R/2$, then instead create an edge from $a$ to $b$ to bring $x_a$ and $x_b$ within a distance $R$. Again, this can be achieved while maintaining $|x_a - x^*| > R/2$. Once the desired radius has been achieved (at a time denoted $t_1$), Lemma \ref{Lemma: bounds on distance travelled} gives that setting controls back to $u^-$ will prevent both $x_a$ and $x_b$ crossing $x^*$ if 
    \begin{equation} \label{Eqn: Additional req on S for a,b}
        \frac{D(t_1)(k_i(t_1) - 1)}{\mathcal{S}} < \frac{R}{4} \,.
    \end{equation}
    As $\mathcal{S} > D(0) \frac{4 (N-1)}{R} $, \eqref{Eqn: Additional req on S for a,b} is satisfied and so there exists a time $T_R \geq T$ at which $x_a(T_R) < x^*$, $x_b(T_R) > x^*$ and $D(T_R)<R$. 

    \textbf{Case 2:} There exists $i'\in\Lambda$ such that $x_{i'}(0) = x^*$.
    
    In this case individuals $a$ and $b$ should be chosen from those $j\in\Lambda$ with $x_j(0)\neq x^*$. This raises the possibility that $|x_a(0) - x_b(0)| > r^*$. As such, once the gathering towards $x_a$ and $x_b$ is complete, both individuals may need to be temporarily connected to $x_{i'}$ to bring the opinion diameter $D$ below $r^*$ and then below $R$, after which point individual $i'$ could simply be guided towards either $x_a$ or $x_b$. This can be done in much the same way as in the previous case, by temporarily connecting either $a$, or $b$, or both, to $i'$ until the desired radius is reached (depending on the distances of $x_a$ and $x_b$ from $x^*$). Indeed, \eqref{Eqn: statement S bound} already ensures that $\mathcal{S}$ is sufficiently large for this to be done without breaking the $r^*$-chain and without $x_a$ and $x_b$ crossing $x^*$. The situation is essentially the same if there are multiple individuals whose initial condition is exactly $x^*$. 

    \textbf{Step 2:} As in Proposition \ref{Proposition: Controllability from empty intial network}, the problem is now reduced to ensuring convergence of $x_a(t)$ and $x_b(t)$ to $x^*$, as the chain of connections is such that all other individuals' opinions tend towards one of these. We begin from a time $T_R$ at which $x_a(T_R) < x^*$, $x_b(T_R) > x^*$ and $D(T_R)<R$. Hence from this time onwards $\phi(|x_j - x_i|) = 1$ for all $i,j\in\Lambda$. Unlike Proposition \ref{Proposition: Controllability from empty intial network}, we cannot now simply reduce to the case $N=2$, as individuals $a$ and $b$ may remain connected to others in the population, albeit with an exponentially decaying weight. 
    
    Note that, up until time $T_R$, $u_{ij} = u^-$ for $i=a,b$ and $j\in\Lambda\setminus\{i\}$.
    Hence \eqref{Eqn: non-crossing x^* ineq} tells us that leaving $u_{ij}$ at $u^-$ will prevent consensus.
    We therefore define times $T_{ab},T_{ba} \in (T_R,\infty)$ such that, for $ij = ab,\,ba$
    \begin{equation} \label{Eqn: Controllability from non-empty network control scheme}
        u_{ij}(t) = 
        \begin{cases}
            u^- & \text{if } 0 \leq t \leq T_{ij} \,,\\
            u^+ & \text{if } t > T_{ij} \,.
        \end{cases}
    \end{equation}
    That is, at times $T_{ij}$ we switch these edges from exponentially decreasing towards 0 to exponentially increasing towards $\ell^+$. Also define
    \begin{equation} \label{Eqn: G definition}
        G(T_{ab},T_{ba}) = \lim_{t\rightarrow\infty} x_a(t) \,,
    \end{equation}
    for the controls given in \eqref{Eqn: Controllability from non-empty network control scheme}. The function $G$ gives the limiting opinion of individual $a$. As $T_{ba} < \infty$ and $w_{bi}$ is exponentially decreasing for all $i\neq a,b$, $G$ will also give the limiting opinion of individual $b$, and thus the location of consensus for the whole population. 
    
    Pick some initial guess for the pair $(\Tilde{T}_{ab}, \Tilde{T}_{ba})$ with both values in the time interval $(T_R,\infty)$. If $G(\Tilde{T}_{ab}, \Tilde{T}_{ba}) > x^*$ then fix $T_{ba} = \Tilde{T}_{ba}$. By Lemma \ref{Lemma: G continuity} in Appendix \ref{Appendix: Proofs}, for a fixed $T_{ba} = \Tilde{T}_{ba} > T_R$, the function $G(T_{ab},T_{ba})$ is continuous in $T_{ab}$. If $T_{ab}$ is made extremely large ($T_{ab}\gg \Tilde{T}_{ba}$), then $x_b$ decreases towards $x_a < x^*$ and, if $T_{ab}$ is sufficiently large, $x_b$ will be sufficiently below $x^*$ before time $T_{ab}$ that $G(T_{ab},\Tilde{T}_{ba}) < x^*$. Hence, by the intermediate value theorem there is a value of $T_{ab}$ for which $G(T_{ab}, \Tilde{T}_{ba}) = x^*$. Due to the persistent (exponentially small) edge weights between individuals $a$ and $b$ and the rest of the population, this ideal value of $T_{ab}$ will need to account for the dynamics of the whole system at $t\rightarrow\infty$. As such its exact value would be impractical to compute in most cases. However, using the continuity of $G$ we have shown that such a value, and therefore such a control, exists. 
    
    If $G(\Tilde{T}_{ab}, \Tilde{T}_{ba}) < x^*$ then use an analogous argument for a fixed $T_{ab} = \Tilde{T}_{ab}$. 
\end{proof}

\begin{remark}
    In the case of an empty starting network, having $x^* = x_{i'}(0)$ for some $i'\in\Lambda$ simplifies the problem significantly. However, in the case of a non-empty starting network, such an $i'$ may instead create additional difficulty as it may form a crucial link in the $r^*$-chain but may repeatedly move above and below $x^*$, making it in some sense unreliable as a point to control towards. 
\end{remark}

Figure \ref{fig: Controllability from non-empty network example} shows a numerical example of the control described in Theorem \ref{Theorem: Controllability from non-empty intial network} using the same setup as described for Figure \ref{fig: Controllability from empty network example}. The initial network $w(0)$ is an Erdos-Renyi random graph with edge probability $p=5/N$ \cite{erdHos1960evolution} with edge weights then given uniformly in the interval $[0,1]$. As previously, suitable values of $\Tilde{T}_{ab}, \Tilde{T}_{ba}$ were identified by iteratively solving the dynamics (this required solving the complete dynamics rather than the $N=2$ case only). Note that we used $s(u_{ij}) = u_{ij}^2$ and so $\mathcal{S}=1$, which is far below the value indicated in \eqref{Eqn: statement S bound} (of approximately 1200), but in this case control is clearly still possible.
\begin{figure}[ht!]
    \centering
    \includegraphics[width = .9\linewidth, trim = {1cm 0cm 1cm 1cm},clip]{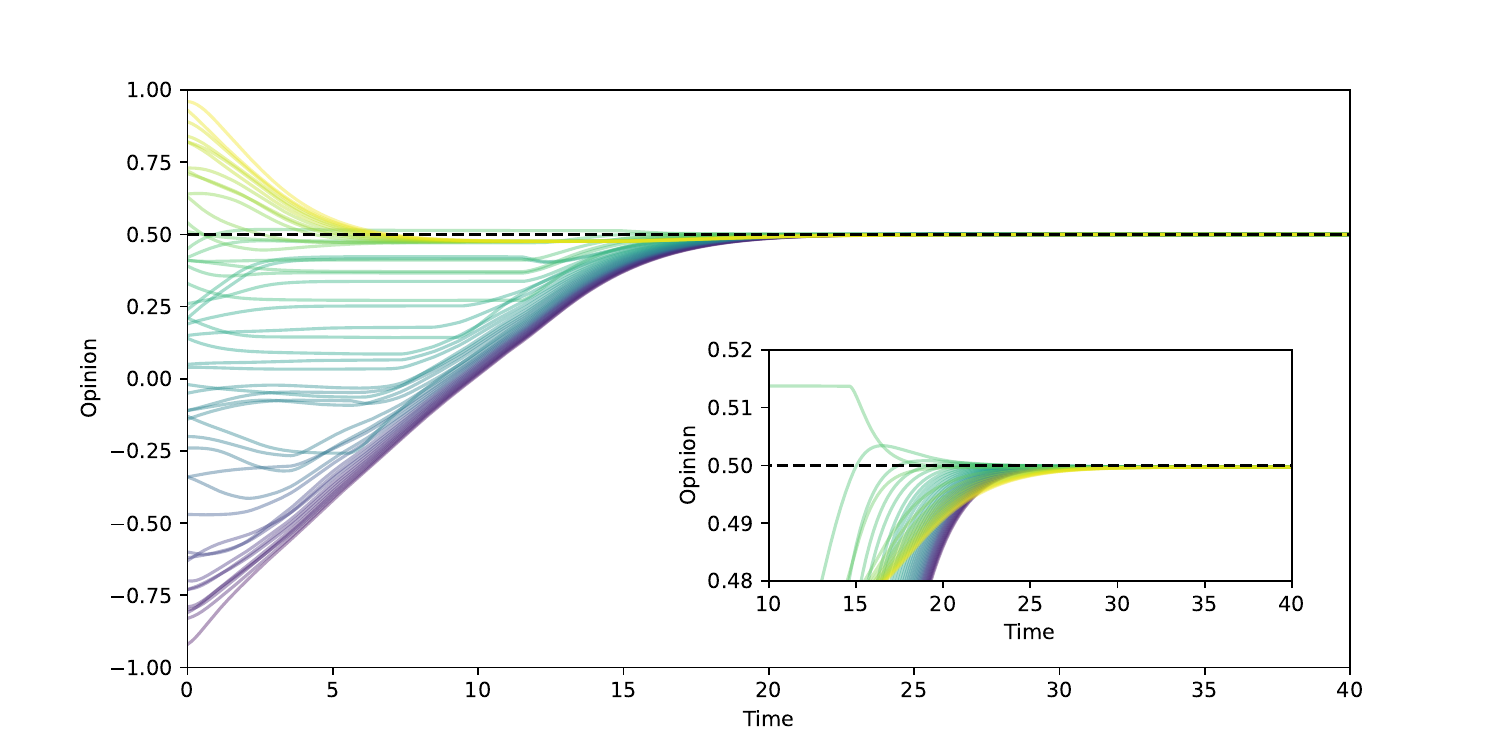}
    \caption{Example of the control method described in Theorem \ref{Theorem: Controllability from non-empty intial network}. Beginning with a non-empty network, edges are created to gather the population closer to individuals near the target opinion $x^* = 0.5$ and removed to prevent crossing this value or splitting the population into multiple clusters. Those individuals closest to $x^*$ are then controlled to consensus precisely at this point. Opinion trajectories are coloured according to individuals' initial opinions. The target opinion $x^*$ is indicated by a black dashed line. The inset plot shows the final approach to consensus at $x^*$.}
    \label{fig: Controllability from non-empty network example}
\end{figure}

One might reasonably expect that consensus at a desired point could be achieved without such drastic controls. The approach described in the proof effectively erases the initial network to ensure that the $r^*$-chain is not broken and the target point remains inside the range of opinions, but it may well be the case that this is unnecessary and indeed, allowing more of the initial network to remain might encourage faster convergence to consensus. Therefore, having established that control is possible, we now aim to improve our control strategy. 

\subsection{Instantaneous control} \label{Section: Instantaneous control}

We next consider an explicit candidate control strategy, inspired by the approach in \cite{piccoli2021control}. In \cite{piccoli2021control}, a model is analysed where each individual has an evolving, non-zero mass that determines their influence, with the total mass preserved across the population. Giving an individual mass $m_j$ is equivalent to setting $w_{ij} = m_j$ for all $i\in\Lambda$ (although the weight dynamics \eqref{eqn: weight ODE} considered in this paper would not preserve the total mass). In the setup in \cite{piccoli2021control} the population always reaches consensus, so it is sufficient to control the location of this consensus. The more general network setting considered in this paper poses additional challenges, but it is possible that a similar approach may provide a viable control strategy. 

Analogously to the mass-weighted mean opinion considered in \cite{piccoli2021control}, we introduce here the degree-weighted mean opinion, given by 
\begin{equation}
    \Bar{x}(t) = \frac{\sum_{i=1}^N k_i^{\text{out}}(t) \, x_i(t)}{\sum_{i=1}^N k_i^{\text{out}}(t)} \,.
\end{equation}
Note that this definition makes use of the out-degree, given by 
\begin{equation} \label{Eqn: out-degree}
    k_i^{\text{out}}(t) = \sum_{j=1}^N w_{ji}(t) \,.
\end{equation}
While the in-degree $k_i$ describes the extent to which individual $i$ is connected to the rest of the population, $k_i^{\text{out}}$ describes the extent to which the population is connected to individual $i$ and thus more closely reflects the idea of individual $i$'s influence and thus their `mass'. As each $k_i^{\text{out}}(t)\geq 1$ for all $t\geq0$, $\Bar{x}(t)$ is always well-defined. 

\begin{proposition} \label{Proposition: Sufficient to control mean}
    If the population reaches consensus, then $\lim_{t\rightarrow\infty} x_i(t) = \lim_{t\rightarrow\infty} \Bar{x}(t)$ for all $i=1,\dots,N$. 
\end{proposition}
\begin{proof}
    Let $x^*$ be the value at which the population reaches consensus, that is for all $i\in\Lambda$, $\lim_{t\rightarrow\infty} x_i(t) = x^*$. 
    Also note that $k_i^{\text{out}}(t)\geq1$, so for all $i\in\Lambda$, $\lim_{t\rightarrow\infty} k_i^{\text{out}}(t) \geq 1 > 0$. 
    Hence 
    \begin{align*}
        \lim_{t\rightarrow\infty} \Bar{x}(t) = \frac{\sum_{i=1}^N \lim\limits_{t\rightarrow\infty} k_i^{\text{out}}(t) \, \lim\limits_{t\rightarrow\infty} x_i(t)}{\sum_{i=1}^N \lim\limits_{t\rightarrow\infty} k_i^{\text{out}}(t)} = \frac{\sum_{i=1}^N \lim\limits_{t\rightarrow\infty} k_i^{\text{out}}(t) \, x^*}{\sum_{i=1}^N \lim\limits_{t\rightarrow\infty} k_i^{\text{out}}(t)} = x^* \,.
    \end{align*}
\end{proof}
Hence controlling the value of the cost function 
\begin{equation} \label{Eqn: Degree-weighted mean}
    V(t) = \frac{1}{2} \big(\Bar{x}(t) - x^* \big)^2 \,
\end{equation}
to $0$ is equivalent to controlling the population to consensus at $x^*$. Differentiating \eqref{Eqn: Degree-weighted mean} gives 
\begin{equation}
    \frac{dV}{dt} = \frac{\big(\Bar{x}(t) - x^* \big)}{\sum_{i=1}^N k_i^{\text{out}}} \, \Bigg( \sum_{i=1}^N \sum_{j=1}^N w_{ij}\,\phi\big(x_j - x_i \big)\,(x_j - x_i) + (x_i - \Bar{x}) \, f(u_{ji},w_{ji}) \Bigg) \,.
\end{equation}
The most negative value of $dV/dt$ is achieved by setting 
\begin{equation} \label{Eqn: instantaneous minimusation of V control strategy}
    u_{ij} = 
    \begin{cases}
        \argmax\limits_{u\in[-M,M]} f(u,w_{ij}) & \text{if } (\Bar{x} - x^*)(\Bar{x} - x_j) > 0 \\
        \argmin\limits_{u\in[-M,M]} f(u,w_{ij}) & \text{if } (\Bar{x} - x^*)(\Bar{x} - x_j) < 0 \,.
    \end{cases}
\end{equation}
That is, we set $u_{ij}$ to be positive for all $j$ if $\Bar{x}$ and $x_j$ are on the same side of $x^*$ and $u_{ij}$ to be negative for all $j$ if $\Bar{x}$ and $x_j$ are on different sides of $x^*$. Note that the control $u_{ij}$ depends on the opinion of individual $i$ only through their contribution to $\Bar{x}$. In addition, it is not possible for $(\Bar{x} - x_2)$ to have the same sign for all $j\in\Lambda$, since $\Bar{x}$ is a convex combination of $x_j$'s, hence there will always be at least one individual $j$ for whom $u_{ij}$ is positive for all $i\in\Lambda$, and thus at least $N-1$ weights must be increasing at any given time. 

To investigate the effectiveness of this control we perform 10,000 simulations in which initial opinions are chosen uniformly at random in the interval $[-1,1]$ and  a weighted Erdos-Renyi random network is created to give the initial network $w(0)$. A consensus target $x^*$ is then chosen uniformly at random inside the initial opinions. Each simulation is run until the opinions have reached a steady state, determined by the distance between opinion vectors at consecutive timepoints falling below a given threshold. For all these simulations we use an exponential interaction function, specifically
\begin{equation} \label{Eqn: exponential interaction function}
    \phi(r) = e^{-r} \,.
\end{equation}
Note that, when the network $w$ is fixed, this interaction function guarantees consensus \cite{nugent2023evolving}. For the control dynamics we take $s$ and $\ell$ given by 
\begin{equation} \label{Eqn: ess and ell}
    s(u) = \mathcal{S} u^2 \,, \quad
    \ell(u) = \frac{1}{2}(u+1) \,,
\end{equation}
and so consider $u_{ij}\in[-1,1]$ (meaning $M=1$). The value of $\mathcal{S}$ determines the speed with which the control acts, and we take $\mathcal{S}=1$ in these simulations. Two example timeseries can be found in Appendix \ref{Appendix: Additional Examples Inst}. From each simulation we ask two questions: has the population reached consensus, and has the population been guided to the consensus target?  

To address the first question we calculate the opinion diameter at the end of each simulation. For all simulations the final opinion diameter was of the order $10^{-5}$, clearly indicating consensus. Results for each simulation, along with a local average, can be found in Figure \ref{fig:final diameter} in Appendix \ref{Appendix: Additional Examples}. The emergence of consensus is perhaps unsurprising. If $\Bar{x}>x^*$ then all edges connecting to individuals with opinions below $\Bar{x}$ will be increasing in weight. Specifically, the edge connecting the individual with the maximum opinion to the individual with the minimum opinion will be increasing in weight. As the interaction function is always positive this will draw the maximum opinion down towards the minimum opinion. If $\Bar{x} < x^*$ then the opposite occurs, yet in either case the opinion diameter is shrinking due to this strengthening of edges. If $\Bar{x} = x^*$ then controls switch off, but we know from \cite{nugent2023evolving} that for an exponential interaction function on a fixed connected network, consensus is guaranteed. As it is highly likely that our randomly generated $w(0)$ will be connected, we thus expect consensus regardless of the value of $x^*$ (although not necessarily at $x^*$). 

To address the question of consensus at $x^*$, we calculate the maximum distance from $x^*$ at the end of the simulation (time $T$), 
\begin{equation*} \label{Eqn: max distance}
    d(x^*) = \max_{i\in\Lambda} | x_i(T) - x^* | \,.
\end{equation*}
Note that the maximum distance cannot exceed $1 + |x^*|$. The results for the 10,000 simulations are shown in Figure \ref{fig:final max distance}. We observe a region of consensus targets, approximately $[-0.5,0.5]$ in which the final maximum distance is extremely small. This indicates that the control is reliably successful for targets in this region. Outside this region the final maximum distance grows almost linearly in $|x^*|$, indicating that the control struggles to achieve consensus above $0.5$ or below $-0.5$. 

\begin{figure}[ht!]
    \centering
    \includegraphics[width = \linewidth]{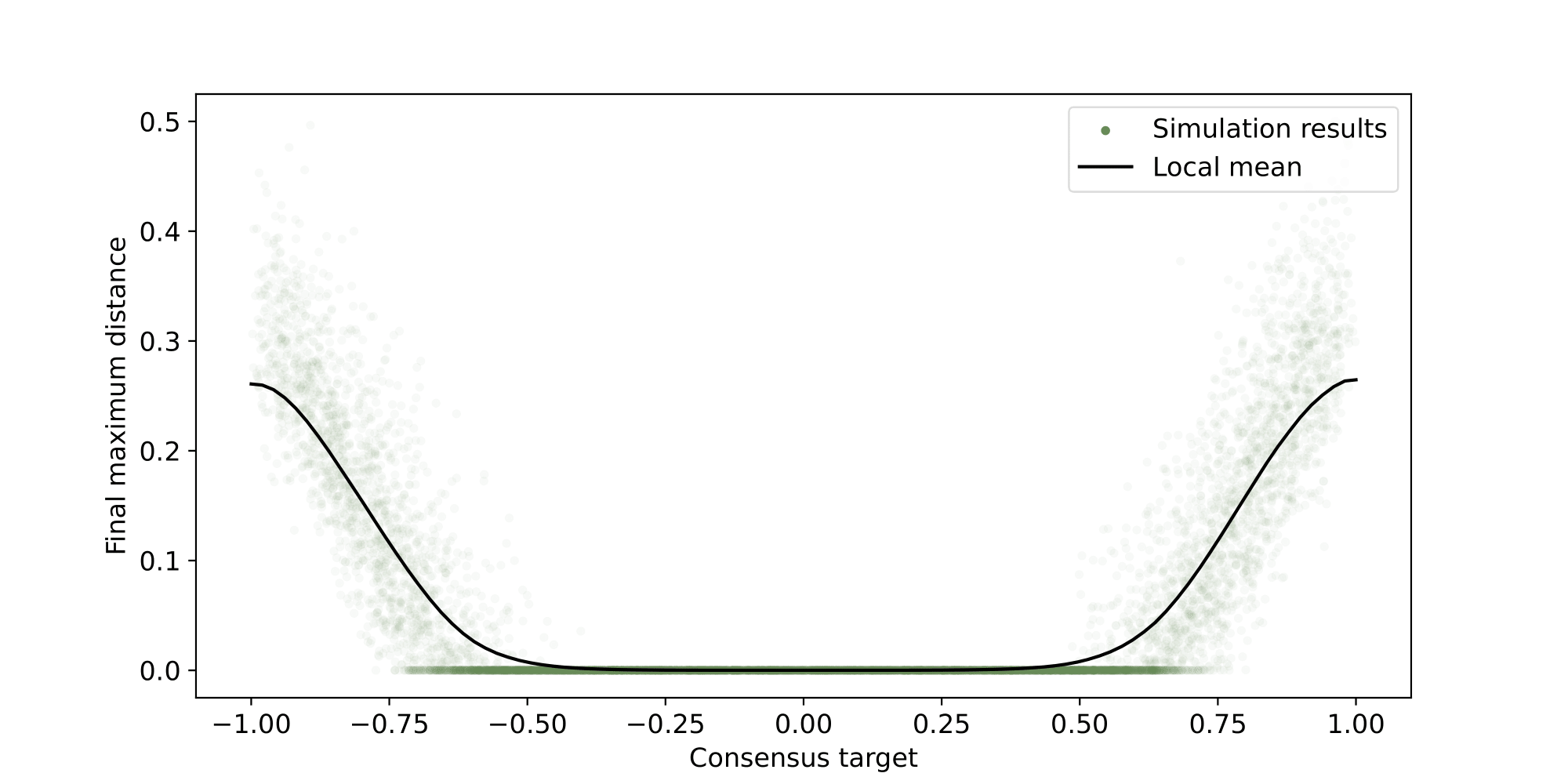}
    \caption{Results of repeated tests of the instantaneous control \eqref{Eqn: instantaneous minimusation of V control strategy} with $M=1$. Each simulation uses uniformly random initial conditions and a weighted Erdos-Renyi random initial network and runs until opinions have reached a steady state. Each point shows the value of the maximum distance from $x^*$, given by \eqref{Eqn: max distance}, at the end of the simulation. The black line shows the local mean.}
    \label{fig:final max distance}
\end{figure}

It is important to note that, although the results in Figure \ref{fig:final max distance} indicate this control strategy is typically successful for $x^*\in[-0.5,0.5]$, this success is not guaranteed. In fact, as previously noted, Proposition \ref{Proposition: Uncontrollability} shows that there exist initial conditions and targets for which this control will fail. As the results presented in Figure \ref{fig:final max distance} were generated using random $x(0)$, $w(0)$ and $x^*$, it is also worth noting that Proposition \ref{Proposition: Uncontrollability} gives a range of initial conditions for which the system is not controllable that has a strictly positive probability of occurring. Thus, if the experiments described were repeated a sufficiently large number of times, we would eventually see some instances with consensus targets in the interval $[-0.5,0.5]$ for which the control is not successful. 

Moreover, this failure cannot necessarily be remedied by increasing $\mathcal{S}$ (the speed with which the control acts). As $w_{ii}(t)=1$ for all $i\in\Lambda$ and $t\geq0$, $\Bar{x}$ cannot take any arbitrary value in the opinion interval, even if the other edge weights could be changed instantaneously. This means that for certain initial conditions control to consensus at $x^*$ is not possible for any $\mathcal{S}$ since, no matter how quickly the control acts, $\Bar{x}$ cannot be controlled to $x^*$ sufficiently quickly. We provide an example of this in Appendix \ref{Appendix: Additional Examples Inst}. 

Having examined the performance of this control strategy for an exponential interaction function, we now consider a smoothed bounded confidence interaction function, under which obtaining consensus is significantly more challenging. Figure \ref{fig:BC disaster} shows that this control is incapable of achieving consensus under such an interaction function. This is due to the fact that the form of the control \eqref{Eqn: instantaneous minimusation of V control strategy} does not include any information about the interaction function and hence cannot take into account the complex behaviours it may cause. For example, in the lower panel of Figure \ref{fig:BC disaster} we observe a situation in which half the population is connected exclusively to individuals outside their confidence bound, meaning they do not interact and their opinions remain constant. Any previous success of this control strategy appears to be heavily reliant on the exponential interaction function allowing interactions at all distances and hence promoting consensus. 

\begin{figure}[ht!]
    \centering
    \includegraphics[width = .8\linewidth]{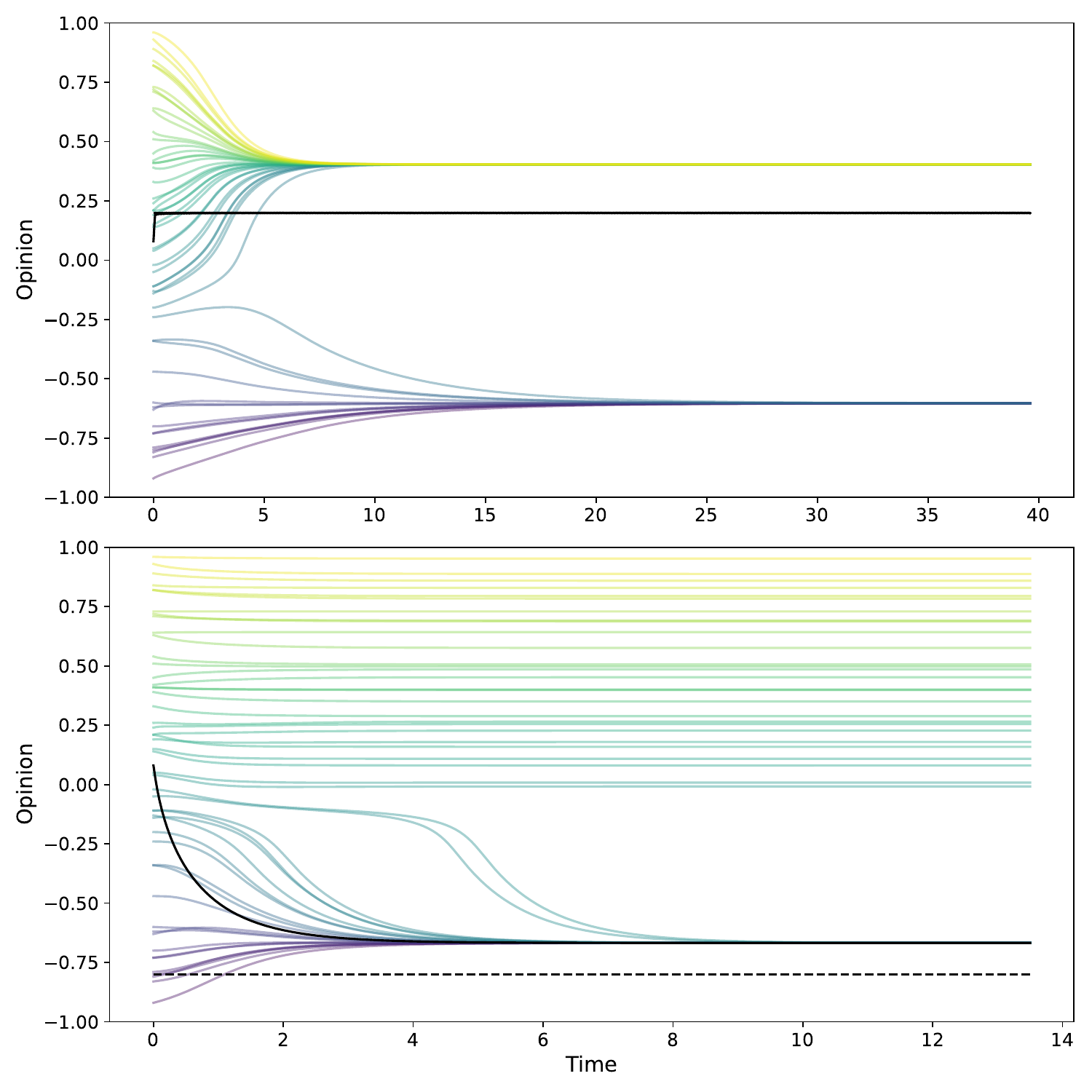}
    \caption{Example implementations of the instantaneous control \eqref{Eqn: instantaneous minimusation of V control strategy} with a smoothed bounded confidence interaction function. Opinion trajectories are coloured according to individuals' initial opinions. This control fails for both the moderate target of $+0.2$ (top panel) and the more extreme target of $-0.5$ (bottom panel) as the population does not reach consensus. The target opinion $x^*$ is indicated by a black dashed line. The degree-weighted mean opinion is given by the solid black line. Note that in the top panel the dashed line is not visible as the degree-weighted mean opinion is almost immediately brought to the target.}
    \label{fig:BC disaster}
\end{figure}

It is possible to help the control account for the interaction function by adapting the cost function $V$ \eqref{Eqn: Degree-weighted mean}. One possibility is to include a term of the form
\begin{equation*}
    \mathcal{E}(t) := \sum_{i,j} w_{ij} \,\varphi\big( |x_i(t) - x_j(t)| \big), \quad \varphi(r) := \int_0^r s \, \phi(s) \,ds \,,
\end{equation*}
as this describes the energy of the system with a fixed network \cite{motsch2014heterophilious,nugent2023evolving}. Minima of $\mathcal{E}$ correspond to stationary opinion states, hence including this term would encourage the control to bring the system towards the nearest stationary state. However, minimising $\mathcal{E}$ does not necessarily encourage consensus. The population may instead be guided rapidly towards a clustered state, preventing consensus and thus preventing consensus at $x^*$. This makes a function involving $\mathcal{E}$ a poor candidate for creating consensus at $x^*$. 

Another solution to the clustering arising from the bounded confidence interaction function is to allow the control more information about the resulting dynamics. In the following Section \ref{Section: Optimal control} we consider an optimal control problem for \eqref{eqn: ODE system} that hopes to remedy the failure observed in Figure \ref{fig:BC disaster} and improve upon the drastic approach in Theorem \ref{Theorem: Controllability from non-empty intial network} of removing the entire initial network. 

\subsection{Optimal Control} \label{Section: Optimal control}

In this section we continue considering controls of the form \eqref{Eqn: Memory weight controls}, with the restriction that, for some fixed $M>0$, $u_{ij}(t)\in[-M,M]$ for all $i,j\in\Lambda$ and $t\geq0$. We denote this set of admissible controls by $\mathcal{U}_{\text{ad}}$. We also consider a finite time horizon, $T > 0$. For a given control $u$ we introduce the cost functional $\mathcal{C}(u)$, which includes both a cost of controls and a cost associated to the distance from the target opinion $x^*$, 
\begin{equation} \label{Eqn: cost functional}
    \mathcal{C}(u) = \int_0^T \alpha \sum_{i=1}^N \sum_{j=1}^N u_{ij}(s)^2 + \beta \sum_{i=1}^N \big(x_i(s) - x^* \big)^2 \,\,ds \,,
\end{equation}
where $(x_i)_{i\in\Lambda}$ is the solution to \eqref{eqn: ODE system} with controls \eqref{Eqn: Memory weight controls}. Including the distance from $x^*$ inside the integral, rather than as a terminal cost, promotes controlling opinions to consensus at $x^*$ as quickly as possible. The constants $\alpha,\beta\in\mathds{R}^+$ are chosen to balance the relative costs of control and distance from $x^*$.

Next we introduce two co-states $p\in\mathds{R}^N$ and $q\in\mathds{R}^{N\times N}$ for $x$ and $w$ respectively. The control theory Hamiltonian \cite{evans1983introduction} associated to \eqref{Eqn: cost functional} is then given by 
\begin{align} 
    H(x,w,p,q,u) &= \sum_{i=1}^N \frac{dx_i}{dt} \, p_i \,+\, \sum_{i=1}^N \sum_{j=1}^N \frac{dw_{ij}}{dt}\,q_{ij} - \alpha \sum_{i=1}^N \sum_{j=1}^N u_{ij}^2 - \beta \sum_{i=1}^N \big(x_i - x^* \big)^2 \,, \\ 
    &= \sum_{i=1}^N \sum_{j=1}^N \bigg( \frac{p_i}{k_i}\,w_{ij}\,\phi(x_j - x_i)(x_j - x_i) + q_{ij}\,s(u_{ij})\big( \ell(u_{ij}) - w_{ij} \big)- \alpha u_{ij}^2 \bigg) - \beta \sum_{i=1}^N \big(x_i - x^* \big)^2 \label{Eqn: Hamiltonian}\,.
\end{align}
Moreover, by the Pontryagin Maximum Principle \cite{evans1983introduction}, the optimal control $\Tilde{u}$ and corresponding states $(\Tilde{x},\Tilde{w})$ and co-states $(\Tilde{p},\Tilde{q})$ satisfy the following 
\begin{enumerate}[label = (\roman*)]
    \item $\dfrac{d\Tilde{x}}{dt} = \nabla_p H(\Tilde{x},\Tilde{w},\Tilde{p},\Tilde{q},\Tilde{u})$ and $\dfrac{d\Tilde{w}}{dt} = \nabla_q H(\Tilde{x},\Tilde{w},\Tilde{p},\Tilde{q},\Tilde{u})$, meaning the original dynamics \eqref{eqn: ODE system} are satisfied. 
    \item $\dfrac{d\Tilde{p}}{dt} = -\nabla_x H(\Tilde{x},\Tilde{w},\Tilde{p},\Tilde{q},\Tilde{u})$ and $\dfrac{d\Tilde{q}}{dt} = -\nabla_w H(\Tilde{x},\Tilde{w},\Tilde{p},\Tilde{q},\Tilde{u})$, referred to as the adjoint dynamics. 
    \item $p(T) = 0$ and $q(T) = 0$, the terminal conditions for $p$ and $q$. 
    \item $H(\Tilde{x},\Tilde{w},\Tilde{p},\Tilde{q},\Tilde{u}) = \max\limits_{u\in\mathcal{U}_{\text{ad}}} H(\Tilde{x},\Tilde{w},\Tilde{p},\Tilde{q},u)$, referred to as the maximisation principle for $\Tilde{u}$.
\end{enumerate}
From the Hamiltonian \eqref{Eqn: Hamiltonian}, the adjoint dynamics are given by 
\begin{subequations} \label{Eqn: Adjoint dynamics}
    \begin{align}
        \frac{dp_i}{dt} &= 2\beta\,(x_i - x^*) - \sum_{j=1}^N \Big( \phi'(x_j - x_i)\,(x_j - x_i) + \phi(x_j - x_i) \Big) \bigg( \frac{w_{ji}}{k_j}\,p_j - \frac{w_{ij}}{k_i}\,p_i \bigg) \,, \label{Eqn: Adjoint p} \\ 
        \frac{dq_{ij}}{dt} &= q_{ij}\,s(u_{ij}) - \frac{p_i}{k_i}\phi(x_j - x_i)\,(x_j - x_i) + \frac{p_i}{k_i^2} \sum_{r=1}^N w_{ir}\,\phi(x_r - x_i)\,(x_r - x_i) \,. \label{Eqn: Adjoint q }
    \end{align}
\end{subequations}

It is not possible to identify a priori the control $\Tilde{u}$ that maximises $H(\Tilde{x},\Tilde{w},\Tilde{p},\Tilde{q},u)$ for general functions $s$ and $\ell$, hence for the remainder of this section we will again work with $s$ and $\ell$ given in \eqref{Eqn: ess and ell}. For these functions an explicit expression for $\Tilde{u}$ can be obtained (see Appendix \ref{Appendix: Max principle} for derivation). Define lower and upper bounds, $b_l$ and $b_u$ respectively, by
\begin{align*}
    b_l(w) = 
    \begin{cases}
        -\infty & \text{ if } w = 0 \\
        -\frac{\alpha}{\mathcal{S}w}  & \text{ if } w > 0
    \end{cases}
    \,,\quad\quad
    b_u(w) = 
    \begin{cases}
        \infty & \text{ if } w = 1 \\
        \frac{\alpha}{\mathcal{S}(1-w)}   & \text{ if } w < 1
    \end{cases}
    \,,
\end{align*}
then the control $\Tilde{u}=(\Tilde{u}_{ij})_{i,j\in\Lambda}$ is given by 
\begin{equation} \label{Eqn: bang bang optimal controls}
    \Tilde{u}_{ij} = 
    \begin{cases}
        -M & \text{if }\, q_{ij} < b_l(w_{ij}) \\
        0 & \text{if }\, b_l(w_{ij}) < q_{ij} < b_u(w_{ij}) \\
        M & \text{if }\,  q_{ij} > b_u(w_{ij}) \,.
    \end{cases}
\end{equation}

It is interesting to note that in both Proposition \ref{Proposition: Controllability from empty intial network}, Theorem \ref{Theorem: Controllability from non-empty intial network} and the optimal control setup \eqref{Eqn: bang bang optimal controls}, the full range of controls offered by \eqref{Eqn: Memory weight controls} is not utilised. Instead only a single value $u^+$ that creates edges and a single value $u^-$ that removes edges are used. This type of control, which switches between extreme values, is commonly known as a bang-bang control. Note that this does not mean edge weights always take integer values, indeed it is often crucial that they do not, only that controls always act to their fullest extent. 

To identify the optimal controls we perform a Forward-Backward Sweep (FBS) over \eqref{eqn: ODE system} and \eqref{Eqn: Adjoint dynamics}, using a 4th order Runga-Kutta numerical scheme, to iteratively improve the controls. Note that during the search for the optimal control we allow continuous controls, rather than limiting to bang-bang controls, but observe convergence towards bang-bang controls of the form \eqref{Eqn: bang bang optimal controls}. We use an initial guess of $u\equiv0$, that is we begin with no control. 

In the following we use the same initial conditions as for Figure \ref{fig: Controllability from non-empty network example} and again take $\mathcal{S}=1$. Additional examples, using different initial conditions, can be found in Appendix \ref{Appendix: Additional Examples}. Our overall aim is to use these examples to make inferences about the general nature of optimal control strategies for this system. 

The trajectory of opinions under the optimal control is shown in Figure \ref{fig:optimal_control_dynamics}. We observe that the trajectory is quite different from that in Figure \ref{fig: Controllability from non-empty network example} as individuals do not wait to be `gathered' towards $x^*$. This leads to significantly faster dynamics that appear to reach consensus before the final time $t=10$, whereas the implementation of the method described in Theorem \ref{Theorem: Controllability from non-empty intial network} does not do so until approximately $t = 30$. Another feature of the dynamics in Figure \ref{fig:optimal_control_dynamics} is that, due to the controls, individuals' opinions rarely cross over. This may help in maintaining the $r^*$-chain and keeping the $x^*$ within the opinion interval. 

\begin{figure}[ht!]
    \centering
    
    \begin{subfigure}{\linewidth}
        \centering
        \includegraphics[width = .8\linewidth, trim = {1cm 0cm 2cm 1cm}, clip]{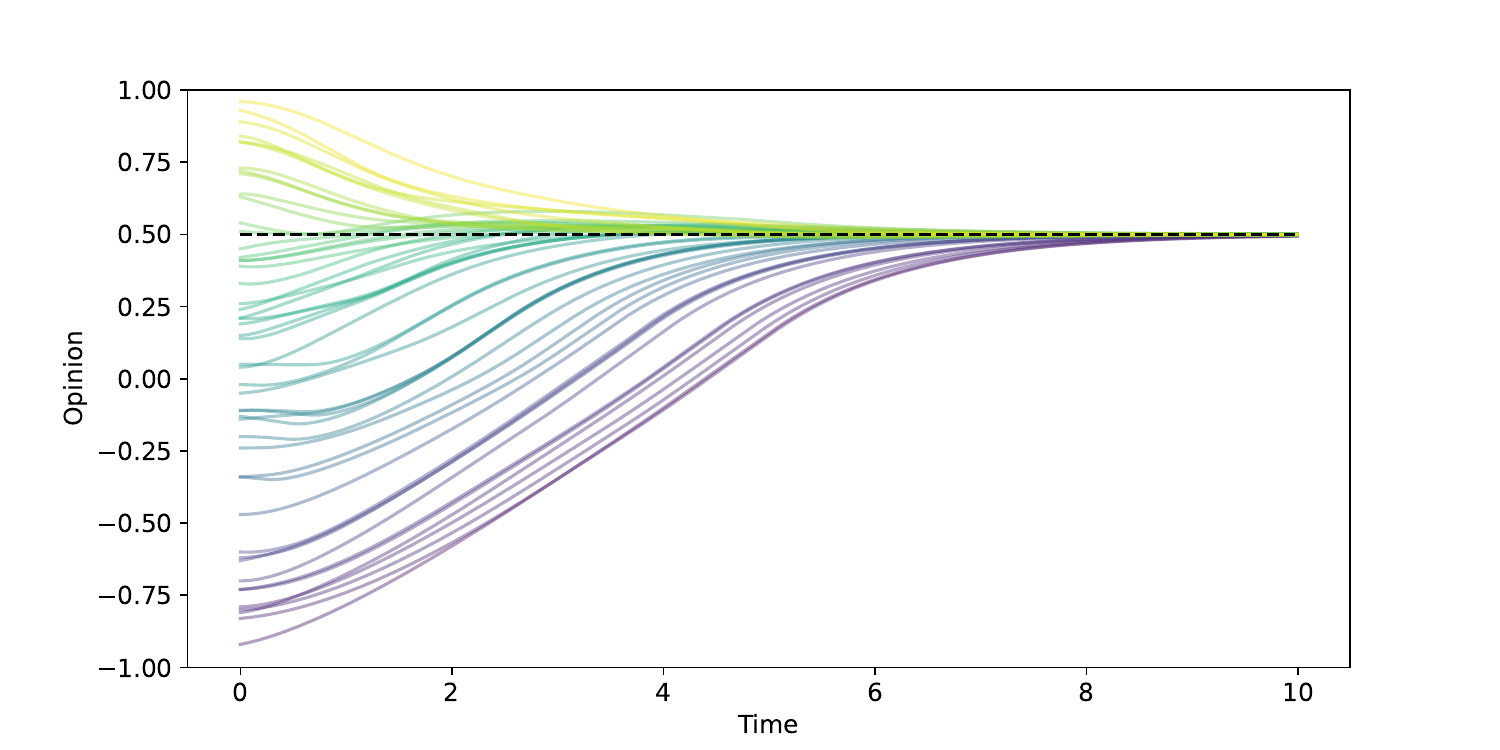}
        \caption{Opinion dynamics under the optimal controls for the cost functional \eqref{Eqn: cost functional}. Opinion trajectories are coloured according to individuals' initial opinions. The target opinion $x^*$ is indicated by a black dashed line.}
        \label{fig:optimal_control_dynamics}
    \end{subfigure}
    
    \begin{subfigure}{\linewidth}
        \centering
        \includegraphics[width = 0.8\linewidth, trim = {0.5cm 0.5cm 1.5cm 1.5cm}, clip]{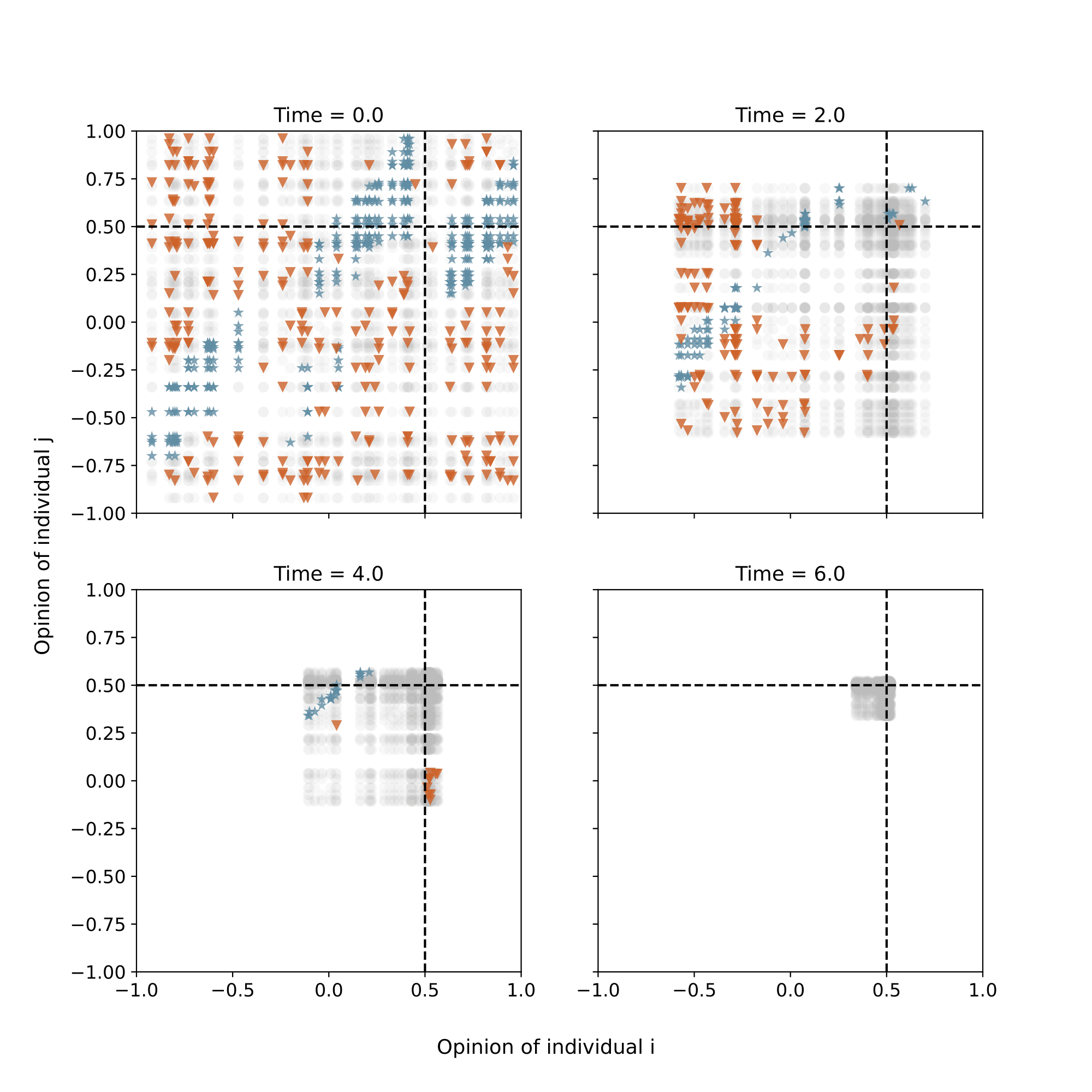}
        \caption{Snapshots of the optimal control at times $t=0,2,4,6$. The horizontal axis gives the opinions $x_i$ for $i\in\Lambda$, the vertical axis gives the opinions $x_j$ for $j\in\Lambda$ and points show the control $u_{ij}$. Blue stars show positive controls, where edges are created/strengthened. Red triangles show negative controls, where edges are being weakened. Grey circles indicate no control. Dashed lines show the location of the target opinion $x^*$, hence as $t$ increases opinions are brought near this value.}
        \label{fig:optimal_control_controls}
    \end{subfigure}
    
    \caption{Results of the FBS to find the optimal controls under \eqref{Eqn: cost functional}, using edge weight dynamics of the form \eqref{Eqn: Memory weight controls} with $s$ and $\ell$ given by \eqref{Eqn: ess and ell}. The same initial conditions were used as for the examples in Figure \ref{fig: Controllability from empty network example} and Figure \ref{fig: Controllability from non-empty network example}.}
    \label{fig:optimal_control_dynamics_and_control}
\end{figure}

Figure \ref{fig:optimal_control_controls} shows the optimal controls at four timepoints ($t = 0,2,4,6$). In each panel we show the value of each $u_{ij}(t)$ in the following way: a point is placed at $(x_i(t), x_j(t))$, that is the horizontal axis corresponds to the opinion of individual $i$ and the vertical axis to the opinion of individual $j$, and this point is coloured according to the value of $u_{ij}(t)$. As our search for the optimal control converges towards a bang-bang control, all control values lie at (or extremely close to) $0$, $+1$ or $-1$, hence the values have been rounded so that a different marker may be used for each value. $0$ appears in grey circles, $+1$ in blue stars and $-1$ in red triangles. 

At time $t=0$ we see a definite striped pattern in the $+1$ controls, showing where edges are being created or strengthened. The width and position of this stripe indicates that edges are strengthened only when pairs of individuals lie within a distance $r^*$ and so could potentially interact. Where $x_i < x^*$ this stripe lies (for the most part) above the diagonal, indicating that individuals are being connected to those with a higher opinion that would bring them up towards $x^*$. When $x_i > x^*$ the opposite occurs. There are some $+1$ controls that do not follow this rule as well as some `missing' points where they might otherwise be expected, suggesting that this typical behaviour is not the only consideration. The distribution of $-1$ controls, showing where edges are being removed, appears to be much more uniform, making the motivation behind removing such edges unclear. 

By time $t=2$ the opinion diameter has reduced and the behaviour of the controls has changed. The stripe of $+1$ controls is still present but has reduced in width as many of the useful edges will have already been established. For $x_i > x^*$ the majority of controls are by now set to zero. There is a new cluster of $-1$ points in the (approximate) region $[-0.6,-0.3]\times[0.5,0.7]$. The purpose of this appears to be to prevent these individuals $i$ with $x_i < x^*$ from increasing above $x^*$, by removing edges connecting them to individuals $j$ with $x_j > x^*$. While it is beneficial to connect individuals to those with higher opinions, in order to bring them closer to $x^*$, the control cannot afford to `overshoot' and encourage individuals to cross $x^*$. 

Moving to time $t=4$ we see that almost all controls are now set to zero. A small stripe of $+1$ controls remains and a new cluster of $-1$ controls has appeared, this time to prevent individuals above $x^*$ from moving below $x^*$. By $t=6$ essentially all controls are set to zero (the very last control switches off at time 7.46). This switching off of controls before the population has neared consensus is made possible by the nature of the forwards-backwards-sweep, in which the control is effectively given knowledge of the future dynamics of the system and thus can adjust weights to account for this. This mirrors the usage of the intermediate value theorem in Proposition \ref{Proposition: Controllability from empty intial network} and Theorem \ref{Theorem: Controllability from non-empty intial network}, where it is shown that the system can be controlled by choosing the correct time to switch on/off controls, but that knowledge of the full future dynamics is needed to identify these times. In a similar way, the optimal control can iteratively update earlier controls to account for the full dynamics, allowing controls to be switched off well before consensus. 

Appendix \ref{Appendix: Additional Examples} contains further examples using the same initial opinions with an empty $w(0)$, complete $w(0)$ and an example using different $x(0)$, $x^*$ and a Watts-Strogatz random network for $w(0)$. In all these examples, similar qualitative behaviours are observed. This provides some support for using examples of optimal control to learn more about what strategies are most effective in general, with the hope that this may pave the way to future analytic results concerning controllability and the optimality of such controls. 

It is clear from Figure \ref{fig:optimal_control_controls} and the examples in Appendix \ref{Appendix: Additional Examples} that the optimal control in this setup is not sparse, meaning that many different edges are controlled throughout the dynamics. In future work we will consider several modifications to the cost function \eqref{Eqn: cost functional}, including methods to encourage a sparse control or penalise deviations from the initial network. The question then shifts from how close the population can be brought to $x^*$, to how much/little control is needed to achieve this. 

\section{Conclusion} \label{Section: Conclusion}

The control problem we have discussed poses several interesting challenges. The difficulty of guaranteeing consensus in opinion formation is exacerbated by the problem of keeping the target point $x^*$ inside the range of current opinions. Controllability can be proven under some rather strong assumptions about the range and speed of controls, yet we observe that successful controls can be found even when these assumptions do not hold. The examples of optimal controls point towards a promising alternative strategy for further theoretical results, as well as the possibility of searching for sparse controls. 

The explicit strategy \eqref{Eqn: instantaneous minimusation of V control strategy} arising from instantaneous control of $V$ showed mixed results, but further work is required to identify precisely where and why its failures occur. One possibility, suggested by the success of optimal control and nature of the controllability results, is that information about the current state only is not sufficient. This raises the question of how much knowledge of the future dynamics a control requires in order to be effective in this setting. 

Recent interest in the extension of mean-field limits to fixed and adaptive networks \cite{gkogkas2021continuum,ayi2024review} also raises the possibility of moving this type of control to the partial differential equation setting. Network control could also be considered at the finer scale of a stochastic agent-based model. Indeed, the type of network control considered here offers many interesting questions about the coupling between opinion and network dynamics and the possibility of successfully influencing their outcome. 

\section{Acknowledgements}

AN was supported by the Engineering and Physical Sciences Research Council through the Mathematics of Systems II Centre for Doctoral Training at the University of Warwick (reference EP/S022244/1). MTW acknowledges partial support from the EPSRC Small Grant EPSRC EP/X010503/1.

For the purpose of open access, the authors have applied a Creative Commons Attribution (CC-BY) license
to Any Author Accepted Manuscript version arising from this submission.

\bibliographystyle{unsrt}
\bibliography{bibliography.bib}

\newpage
\appendix
\counterwithin*{equation}{section}
\renewcommand\theequation{\thesection\arabic{equation}}

\section{Proofs} \label{Appendix: Proofs}

\begin{lemma} \label{Lemma: F continuity}
    Fix some $T_{21}>0$, then the function $F(T_{12},T_{21})$ for $T_{12}\in\mathds{R}^+$ defined by \eqref{Eqn: F definition} in Proposition \ref{Proposition: Controllability from empty intial network} is continuous in $T_{12}$. 
\end{lemma}
\begin{proof}
    Fix a $\tau \geq 0$. Consider two versions of the dynamics, the first $(x^{(1)},w^{(1)})$ in which $T_{12} = \tau$ and the second $(x^{(2)},w^{(2)})$ in which $T_{12} = \tau + h$, for some small $h > 0$. Both versions are identical until time $\tau$, so we consider only $t\geq\tau$. As both versions have the same fixed $T_{21}$ we have
    \begin{align*}
        w^{(1)}_{21}(t) = w^{(2)}_{21}(t) \,. 
    \end{align*}
    for all $t\geq\tau$. Moreover, given the values of $T_{12}$ we also have
    \begin{align} 
        w^{(1)}_{12}(t) &= w^{(1)}_{12}(\tau) = \ell^+ (1 - e^{-s^+ \tau} ) \,, \label{Eqn: Lemma A1 w_12 v1} \\[0.4em]
        w^{(2)}_{12}(t) &=
        \begin{cases}
            \ell^+ (1 - e^{-s^+ t} ) & \text{if } \tau \leq t < \tau + h \,, \\
            \ell^+ (1 - e^{-s^+ (\tau + h)} ) & \text{if } t \geq \tau + h \,.
        \end{cases} \label{Eqn: Lemma A1 w_12 v2} 
    \end{align}
    Hence there exists a positive function $\gamma:[\tau,\infty)\rightarrow\mathds{R}^+$ such that
    \begin{equation} \label{Eqn: Lemma A1, difference in w12's}
        w^{(2)}_{12}(t) = w^{(1)}_{12}(t) + \gamma(t) \,.
    \end{equation}

    Define $z^{(v)}(t) = x_2^{(v)}(t) - x_1^{(v)}(t) \geq 0$ for $v=1,2$. Then
    \begin{align*}
        \frac{dz^{(v)}}{dt} 
        &= \bigg(\frac{w_{21}^{(v)}}{1 + w_{21}^{(v)}}\bigg) \phi\big( | x_1^{(v)} - x_2^{(v)}| \big) \big( x_1^{(v)} - x_2^{(v)} \big) - \bigg(\frac{w_{12}^{(v)}}{1 + w_{12}^{(v)}}\bigg) \phi\big( | x_2^{(v)} - x_1^{(v)}| \big) \big( x_2^{(v)} - x_1^{(v)} \big) \,,\\
        &= - \bigg( \frac{w_{21}^{(v)}}{1 + w_{21}^{(v)}} + \frac{w_{12}^{(v)}}{1 + w_{12}^{(v)}} \bigg) \phi\big(|z^{(v)}|\big) \, z^{(v)} \,,\\[0.4em]
        &= - \Omega^{(v)}(t) \, \phi\big(|z^{(v)}|\big) \, z^{(v)} \,,
    \end{align*}
    where the function $\Omega^{(v)}(t)$ is given by
    \begin{equation*}
        \Omega^{(v)}(t) =  \frac{w_{21}^{(v)}(t)}{1 + w_{21}^{(v)}(t)} + \frac{w_{12}^{(v)}(t)}{1 + w_{12}^{(v)}(t)} \,.
    \end{equation*}
    As $\tau > 0$ and all $w_{ij}$'s are non-decreasing, $\Omega^{(v)}(t) \geq \Omega^{(v)}(\tau)$. Furthermore, $z^{(v)}$ is decreasing and $z^{(v)}(0) < r^*$, hence $\Omega^{(v)}(t) \, \phi\big(|z^{(v)}|\big) \geq \Omega^{(v)}(\tau) \,\phi\big(|z^{(v)}(0)|\big) > 0$. Thus there exists a positive constant $c_1$ (independent of the version $(v)$) such that, for all $t\geq\tau$, $v=1,2$
    \begin{align} \label{Eqn: Lemma A1, exponential bound on z}
        z^{(v)}(t) \leq z^{(v)}(\tau)  \, e^{ c_1 (t - \tau) } \,.
    \end{align}
    Note that if $\tau = 0$ a slight modification is required as $\Omega^{(v)}(\tau) = 0$, but a bound of the same form can still be obtained. 
    
    Additionally, by \eqref{Eqn: Lemma A1, difference in w12's} we can write $\Omega^{(2)}(t) = \Omega^{(1)}(t) + \Gamma(t)$ for a positive function $\Gamma:[\tau,\infty)\rightarrow\mathds{R}^+$. Due to this, $z^{(1)}$ and $z^{(2)}$ differ only in a rescaling of time. Specifically we can write 
    \begin{align} \label{Eqn: Lemma A1 time change for z}
        z^{(2)}(t) = z^{(1)}\big(t + \delta(t)\big) \,,
    \end{align}
    for 
    \begin{align*}
        \delta(t) = \int_\tau^t \frac{\Gamma(r)}{\Omega^{(1)}(r)} \, dr \,.
    \end{align*}
    Using \eqref{Eqn: Lemma A1 w_12 v1} and \eqref{Eqn: Lemma A1 w_12 v2} it can be verified that $\delta(t) \leq m(h)\,t$ for some function $m(h)$ with $m(h)\rightarrow0$ as $h\rightarrow0$. 
    
    Combining this time rescaling with \eqref{Eqn: Lemma A1, exponential bound on z} we obtain, for any $t\geq\tau$, 
    \begin{align}
        \big|z^{(2)}(t) - z^{(1)}(t)\big| 
        &= \big|z^{(1)}\big(t + \delta(t)\big) - z^{(1)}(t)\big| \nonumber\\[0.4em]
        &= \bigg|\int_t^{t + \delta(t)} \Omega^{(v)}(r) \, \phi\big(|z^{(v)}(r)|\big) \, z^{(v)}(r) \,dr \,\bigg| \nonumber\\
        &\leq \int_t^{t + \delta(t)} z^{(v)}(r) \,dr \nonumber\\
        &\leq \int_t^{t + \delta(t)} z^{(v)}(\tau) \, \, e^{ c_1 (r - \tau) } \,dr \nonumber\\[0.4em]
        &= z^{(v)}(\tau) \, e^{ c_1 \tau } \frac{1}{c_1} \big( e^{ - c_1 t } - e^{ - c_1 (t+\delta(t)) } \big) \nonumber\\
        &\leq c_2 \, e^{ - c_1 t } \big( 1 - e^{ - c_1 m(h) \, t } \big) \nonumber\\[0.4em]
        &\leq c_2 \, m(h)\,t\, e^{ - c_1 t } \label{Eqn: Lemma A1 bound on z difference}
    \end{align}
    where the constant $c_2$ is given by $c_2:= z^{(v)}(\tau) \, e^{ c_1 \tau } \frac{1}{c_1}$. 

    We can now obtain a uniform-in-time estimate for the difference between $x_1^{(1)}(t)$ and $x_1^{(2)}(t)$, allowing us to then considering the difference between their limits at $t\rightarrow\infty$. 
    \begin{align*}
        | x_1^{(2)}(t) - x_1^{(1)}(t) | &= \bigg| \int_\tau^t \bigg(\frac{w_{12}^{(2)}}{1 + w_{12}^{(2)}}\bigg) \phi\big( | z^{(2)} | \big) \, z^{(2)} - \bigg(\frac{w_{12}^{(1)}}{1 + w_{12}^{(1)}}\bigg) \phi\big( | z^{(1)} | \big) \, z^{(1)} \,dr \bigg| \\
        &\leq \int_\tau^t \bigg| \frac{w_{12}^{(2)}}{1 + w_{12}^{(2)}} - \frac{w_{12}^{(1)}}{1 + w_{12}^{(1)}} \bigg| \, \phi\big( | z^{(2)} | \big) \, z^{(2)} \, + \bigg(\frac{w_{12}^{(1)}}{1 + w_{12}^{(1)}}\bigg) \Big| \phi\big( | z^{(2)} | \big) \, z^{(2)} - \phi\big( | z^{(1)} | \big) \, z^{(1)} \Big|\,dr \\
        &\leq \Bigg| \frac{w_{12}^{(2)}(\tau  +h)}{1 + w_{12}^{(2)}(\tau  +h)} - \frac{w_{12}^{(2)}(\tau)}{1 + w_{12}^{(2)}(\tau)} \Bigg| \, \int_\tau^t \, z^{(2)} \, dr + \frac{1}{2}(1 + L_\phi)\, \int_\tau^t \big| z^{(2)}(r) - z^{(1)}(r) \big|\,dr \\
    \end{align*}
    where $L_\phi$ is the Lipschitz constant for $\phi$. As both $z^{(2)}$ and $\big| z^{(2)}(r) - z^{(1)}(r) \big|$ are exponentially decreasing their integrals are bounded as $t\rightarrow\infty$. In addition, by \eqref{Eqn: Lemma A1 bound on z difference}
    \begin{equation*}
        \int_\tau^\infty \big| z^{(2)}(r) - z^{(1)}(r) \big|\,dr \leq c_3 \, m(h) \,,
    \end{equation*}
    for some positive constant $c_3$. Furthermore, as $w_{12}^{(2)}$ is continuous, 
    \begin{equation*}
        \lim_{h\rightarrow0} \,\Bigg| \frac{w_{12}^{(2)}(\tau  +h)}{1 + w_{12}^{(2)}(\tau  +h)} - \frac{w_{12}^{(2)}(\tau)}{1 + w_{12}^{(2)}(\tau)} \Bigg| = 0 \,.
    \end{equation*}
    Overall this gives that $| x_1^{(2)}(t) - x_1^{(1)}(t) |$ is bounded above, uniformly in time, with a bound that tends to $0$ as $h\rightarrow\infty$. Thus for $h>0$
    \begin{equation*}
        \lim_{h\rightarrow0} F(\tau + h, T_{21}) = F(\tau, T_{21}) \,.
    \end{equation*}
    By an almost identical argument, reversing the time change in \eqref{Eqn: Lemma A1 time change for z}, the same holds for $h < 0$.  Hence we can conclude that, for a fixed $T_{21} > 0$, $F(T_{12}, T_{21})$ is continuous in $T_{12}$. 
\end{proof}

\begin{lemma} \label{Lemma: G continuity}
    Fix some $T_{ba}>T$, as defined in Theorem \ref{Theorem: Controllability from non-empty intial network}. Then the function $G(T_{ab},T_{ba})$ for $T_{ab}\in(T,\infty)$, defined by \eqref{Eqn: G definition}, is continuous in $T_{ab}$. 
\end{lemma}
\begin{proof}
    As weights cannot be driven to exactly zero, only made exponentially small, we do not look specifically at the location of $\lim_{t\rightarrow\infty} x_a(t)$, but instead consider the behaviour of the entire ODE system. We show that making a small change to $T_{ab}$, and therefore a small change to the weight $w_{ab}$ has a correspondingly small impact on the location of $\lim_{t\rightarrow\infty} x(t)$ as a whole, and thus on the location of $\lim_{t\rightarrow\infty} x_a(t)$. 

    Recall that for $t\geq T$, $\phi\big(|x_j(t) - x_i(t)|\big) = 1$ for all $i,j\in\Lambda$. Hence the ODE system \eqref{eqn: ODE system} becomes
    \begin{equation*}
        \frac{dx_i}{dt} = \frac{1}{k_i} \sum_{j=1}^N w_{ij} \,(x_j - x_i) = \Bigg( \sum_{j\neq i} \frac{w_{ij}}{k_i} \, x_j \Bigg) - \Bigg( \sum_{j\neq i} \frac{w_{ij}}{k_i} \Bigg) \, x_i  \,.
    \end{equation*}
    The system can then be written as 
    \begin{equation*}
        \frac{dx}{dt} = A(t) \, x \,,
    \end{equation*}
    for the matrix
    \begin{equation*}
        A_{ij} = 
        \begin{cases}
            \dfrac{w_{ij}}{k_i} \,&\text{for } i\neq j \,,\\[0.5em]
            - \sum\limits_{j\neq i} \dfrac{w_{ij}}{k_i} \,&\text{for } i=j \,.\\
        \end{cases}
    \end{equation*}
    The solution to this system can then be written in the form $x(t) = x(0)\,\exp\Big( \int_0^t A(r) \,dr \Big)$. 

    As in the proof of Lemma \ref{Lemma: F continuity} we consider two versions of the system, the first $(x^{(1)},w^{(1)})$ in which $T_{ab} = \tau$ and the second $(x^{(2)},w^{(2)})$ in which $T_{ab} = \tau + h$, for some small $h > 0$. Define $\gamma(t) = w_{ab}^{(1)}(t) - w_{ab}^{(2)}(t)$. The controls $u_{ab}$ are known in both versions, hence we can compute $\gamma(t)$ exactly. For $t\leq\tau$, $\gamma(t) = 0$. Then for $\tau < t < \tau + h$, 
    \begin{equation*}
        \gamma(t) = \ell^+ - \big(\ell* - w_{ab}(\tau) \big) \, e^{-s^+(t - \tau)} - w_{ab}(\tau) \, e^{-\mathcal{S}(t - \tau)} \,.
    \end{equation*}
    Finally, for $t \geq \tau + h$, 
    \begin{equation*}
        \gamma(t) = e^{-s^+ (t-\tau)} \bigg( \ell^+ \big(e^{s^+h} - 1\big) + w_{ab}(\tau)\Big( 1 - e^{(s^+ - \mathcal{S})h} \Big) \bigg) \,.
    \end{equation*}
    Crucially,
    \begin{equation} \label{Eqn: Lemma A2, gamma integral bound}
        \int_\tau^\infty | \gamma(t) | \, dt < m(h)
    \end{equation}
    where $m(h):\mathds{R}\rightarrow\mathds{R}$ is a bounded function with $m(h)\rightarrow0$ as $h\rightarrow0$. 

    This change in $w_{ab}$ translates into a small alteration to the matrix $A(t)$. Let $A^{(1)}(t) = A^{(2)}(t) + \Gamma(t)$. This gives that 
    \begin{align}
        \| x^{(1)}(t) - x^{(2)}(t) \| &= \bigg\| x(0)\,\exp\bigg( \int_0^t A^{(1)}(r) \,dr \bigg) - x(0)\,\exp\bigg( \int_0^t A^{(2)}(r) \,dr \bigg) \bigg\| \nonumber\\
        &= \bigg\| x(0)\,\exp\bigg( \int_0^t A^{(2)}(r) + \Gamma(r) \,dr \bigg) - x(0)\,\exp\bigg( \int_0^t A^{(2)}(r) \,dr \bigg) \bigg\| \nonumber\\
        &= \bigg\| x(0)\,\exp\bigg( \int_0^t A^{(2)}(r) \,dr \bigg) \exp\bigg( \int_0^t \Gamma(r) \,dr \bigg) - x(0)\,\exp\bigg( \int_0^t A^{(2)}(r) \,dr \bigg) \bigg\| \nonumber\\
        &\leq \bigg\| x(0)\,\exp\bigg( \int_0^t A^{(2)}(r) \,dr \bigg) \bigg\| \times \bigg\|\exp\bigg( \int_0^t \Gamma(r) \,dr \bigg) - \mathds{1}_{N \times N} \bigg\|_{op} \nonumber\\
        &= \| x^{(2)}(t)  \| \times \bigg\|\exp\bigg( \int_0^t \Gamma(r) \,dr \bigg) - \mathds{1}_{N \times N} \bigg\|_{op} \label{Eqn: Lemma A2, operator norm bound}
    \end{align}
    where $\mathds{1}_{N \times N}$ is the $N\times N$ identity matrix and $\|\cdot\|_{op}$ is the operator norm. 
    
    We will now bound the entries of $\int_0^t \Gamma(r) \,dr$. As the only change between the two versions is in $T_{ab}$ and therefore in $w_{ab}$, $\Gamma_{ij} = 0$ for $i\neq a$. For $i=a$,
    \begin{align*}
        \Gamma(t)_{ab} 
        &= \frac{w_{ab}^{(2)}(t) + \gamma(t)}{k_a^{(2)}(t) + \gamma(t)} - \frac{w_{ab}^{(2)}(t)}{k_a^{(2)}(t)} \\
        &= \frac{1}{k_a^{(2)}(t) + \gamma(t)} \bigg( 1 - \frac{1}{k_a^{(2)}(t)} \bigg)\,\gamma(t) \\
        &\leq \gamma(t) \,,
    \end{align*}
    and for $j\neq a,b$
    \begin{align*}
        \Gamma(t)_{aj} 
        &= \frac{w_{aj}^{(2)}(t)}{k_a^{(2)}(t) + \gamma(t)} - \frac{w_{aj}^{(2)}(t)}{k_a^{(2)}(t)} \\
        &= \frac{w_{aj}^{(2)}(t)}{k_a^{(2)}(t) + \gamma(t)} \bigg( 1 - \frac{k_a^{(2)}(t) + \gamma(t)}{k_a^{(2)}(t)} \bigg) \\
        &= \frac{w_{aj}^{(2)}(t)}{k_a^{(2)}(t) + \gamma(t)} \gamma(t) \\
        &\leq \gamma(t) \,.
    \end{align*}
    Finally,
    \begin{align*}
        \big|\Gamma(t)_{aa}\big| = \bigg|- \sum_{j\neq a} \Gamma_{aj}(t) \bigg| \leq (N-1)\,\gamma(t) \,.
    \end{align*}
    Therefore by \eqref{Eqn: Lemma A2, gamma integral bound}, for any $h>0$,
    \begin{equation*}
        \int_0^\infty \Gamma(r) \,dr < \infty \,,
    \end{equation*}
    and 
    \begin{equation} \label{Eqn: Lemma A2, Gamma intergral tends to zero}
        \lim_{h\rightarrow0} \int_0^\infty \Gamma(r) \,dr = 0 \,.
    \end{equation}
    As both the operator norm and matrix exponential are continuous, combining \eqref{Eqn: Lemma A2, operator norm bound} and \eqref{Eqn: Lemma A2, Gamma intergral tends to zero} gives that 
    \begin{equation}
        \Big\| \lim_{t\rightarrow\infty} x^{(1)}(t) - \lim_{t\rightarrow\infty} x^{(2)}(t) \Big\| \rightarrow 0 
    \end{equation}
    as $h\rightarrow0$. An almost identical argument holds for small $h < 0$ (simply consider $-\gamma(t)$). Hence for a fixed $T_{ba}$, the function $G(T_{ab},T_{ba})$ is continuous in $T_{ab}$. 
\end{proof}

\section{Maximisation Principle} \label{Appendix: Max principle}

Taking $s$ and $\ell$ as defined in \eqref{Eqn: ess and ell}, the Hamiltonian \eqref{Eqn: Hamiltonian} becomes
\begin{align*} 
    H(x,w,p,q,u) 
    &= \sum_{i=1}^N \sum_{j=1}^N \bigg( \frac{p_i}{k_i}\,w_{ij}\,\phi(x_j - x_i)(x_j - x_i) + q_{ij}\,\mathcal{S}u_{ij}^2 \big( (u_{ij}+1)/2 - w_{ij} \big)- \alpha u_{ij}^2 \bigg) \\ &- \beta \sum_{i=1}^N \big(x_i - x^* \big)^2 \,.
\end{align*}
When trying to maximise this function with respect to $u$, we can consider the effect of each $u_{ij}$ independently. Furthermore, much of $H$ is independent of the choice of $u$. Hence we define a new function
\begin{equation}
    h(\upsilon \,;q,w) = q \,\mathcal{S}\upsilon^2 \big( (\upsilon+1)/2 - w \big)- \alpha \upsilon^2 \,,
\end{equation}
and set
\begin{equation}
    u_{ij} = \argmax_{\upsilon\in[-1,1]} \, h(\upsilon, q_{ij}, w_{ij}) \,.
\end{equation}
If $\mathcal{S}q = 0$ then clearly $h$ has its maximum at $\upsilon = 0$, hence we assume this is not the case. The function $h$ can then be written (dropping the explicit dependence on $q$ and $v$)
\begin{align*}
    h(\upsilon) &= \frac{\mathcal{S}q}{2} \upsilon^2 \bigg( \upsilon + \Big(1 - 2w - \frac{2\alpha}{\mathcal{S}q} \Big) \bigg) \,, \\
    &= a \upsilon^2 ( \upsilon + b) \,,
\end{align*}
for $a = \frac{\mathcal{S}q}{2}$ and $b = \Big(1 - 2w - \frac{2\alpha}{\mathcal{S}q} \Big)$, both of which can take positive or negative values. The local extrema of $h$ are at $\upsilon = 0$ and $\upsilon^* := -\frac{2b}{3}$, with $h(0)=0$ and $h(\upsilon^*) = \frac{4ab^3}{27}$.

Recall that $w\in[0,1]$, $\alpha > 0$, $\mathcal{S}>0$ and $q\in\mathds{R}$. We aim to show that under these conditions, the maximum of $h(\upsilon) $ over $\upsilon\in[-1,1]$ is never at $\upsilon = \upsilon^*$, hence must be in the set $\{-1,0,1\}$. We therefore assume that there exists $w,\alpha,\mathcal{S},q$ satisfying the above conditions for which the maximum does indeed lie at $\upsilon^*$ and aim to find a contradiction.

As the maximum lies at $\upsilon^*$ we have the following inequalities
\begin{align*}
    h(\upsilon^*) &\geq h(0) = 0 \\
    h(\upsilon^*) &\geq h(1) = a(b+1) \\
    h(\upsilon^*) &\geq h(-1) = a(b-1)
\end{align*}
These are all satisfied only when $a\geq 0$ and $b\geq 3$, or when $a\leq0$ and $b\leq-1$. Note that the sign of $a$ coincides with the sign of $q$ and $q\neq0$, so these can be rewritten as $q> 0$ and $b\geq 3$, or $q<0$ and $b\leq-1$.

However, if $q>0$ then $b\leq 1 - 2w \leq 1 < 3$, so the first pair of inequalities is not possible. Also, if $q<0$ then $b> 1-2w \geq -1$, so the second pair of inequalities is not possible. 

Overall this means the maximum of $h(\upsilon) $ over $\upsilon\in[-1,1]$ is always in the set $\{-1,0,1\}$. Comparing the values of $h(-1), h(0), h(1)$ gives \eqref{Eqn: bang bang optimal controls}.

\section{Additional Examples of Instantaneous Control} \label{Appendix: Additional Examples Inst}

Figure \ref{fig: Instantaneous control of V} demonstrates examples in which the control strategy \eqref{Eqn: instantaneous minimusation of V control strategy} succeeds in guiding the population to consensus at a moderate value of $x^*=0.2$, but fails for the more extreme value $x^*=-0.8$. Both examples use the setup described in Section \ref{Section: Instantaneous control}, but with identical initial conditions. 

\begin{figure}[ht!]
    \centering
    \includegraphics[width = .8\linewidth]{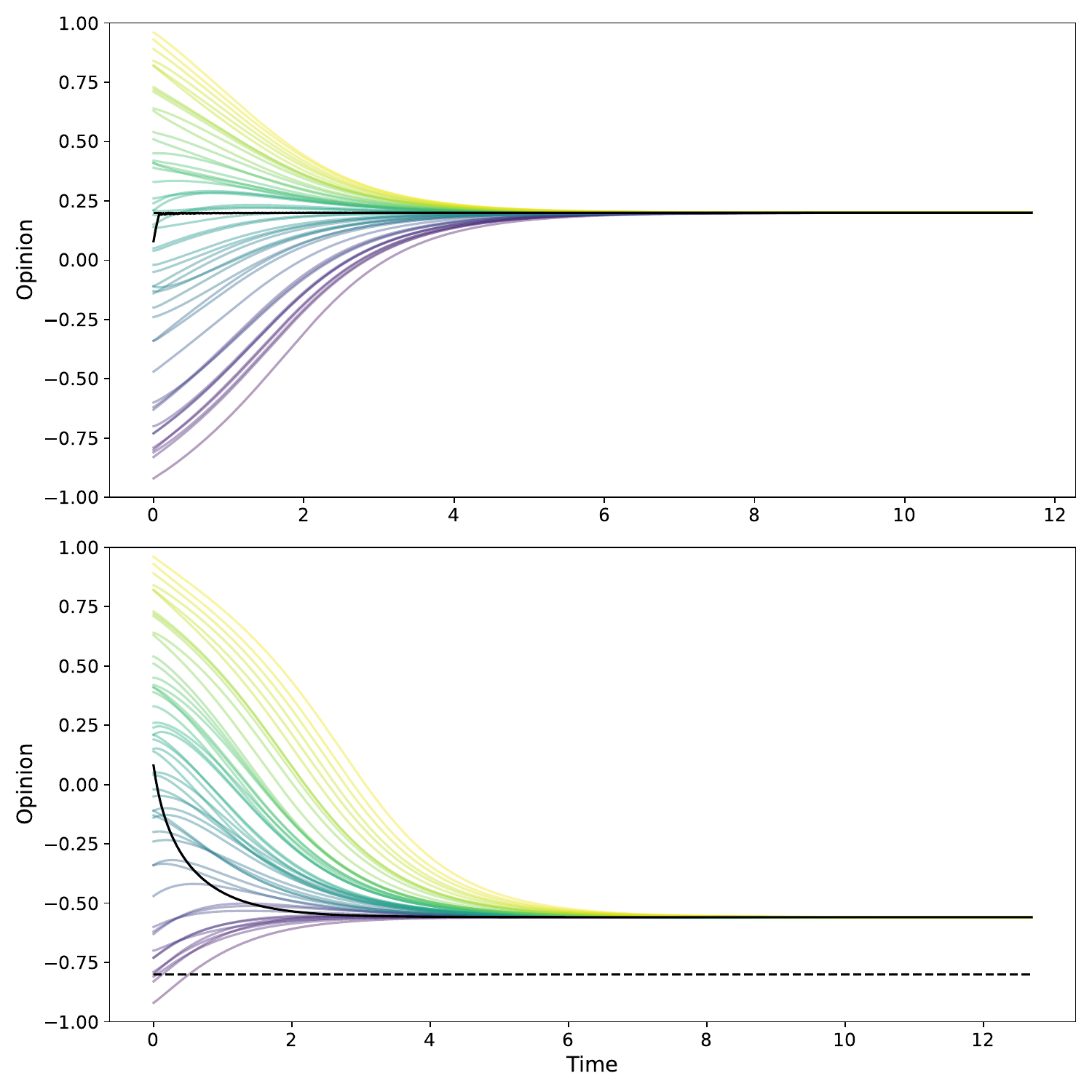}
    \caption{Example implementations of the instantaneous control \eqref{Eqn: instantaneous minimusation of V control strategy} with exponential interaction function \eqref{Eqn: exponential interaction function}. This control succeeds for the moderate target of $+0.2$ (top panel), but fails for the more extreme target of $-0.8$ (bottom panel). Opinion trajectories are coloured according to individuals' initial opinions. The target opinion $x^*$ is indicated by a black dashed line. The degree-weighted mean opinion is given by the solid black line. Note that in the top panel the dashed line is not visible as the degree-weighted mean opinion is almost immediately brought to the target.}
    \label{fig: Instantaneous control of V}
\end{figure}

Figure \ref{fig:final diameter} shows the final opinion diameter in each of the 10,000 simulations described in Section \ref{Section: Instantaneous control}. In all cases the population can be considered to have reached consensus as the opinion diameter is of the order $10^{-5}$. 

\begin{figure}[ht!]
    \centering
    \includegraphics[width = \linewidth]{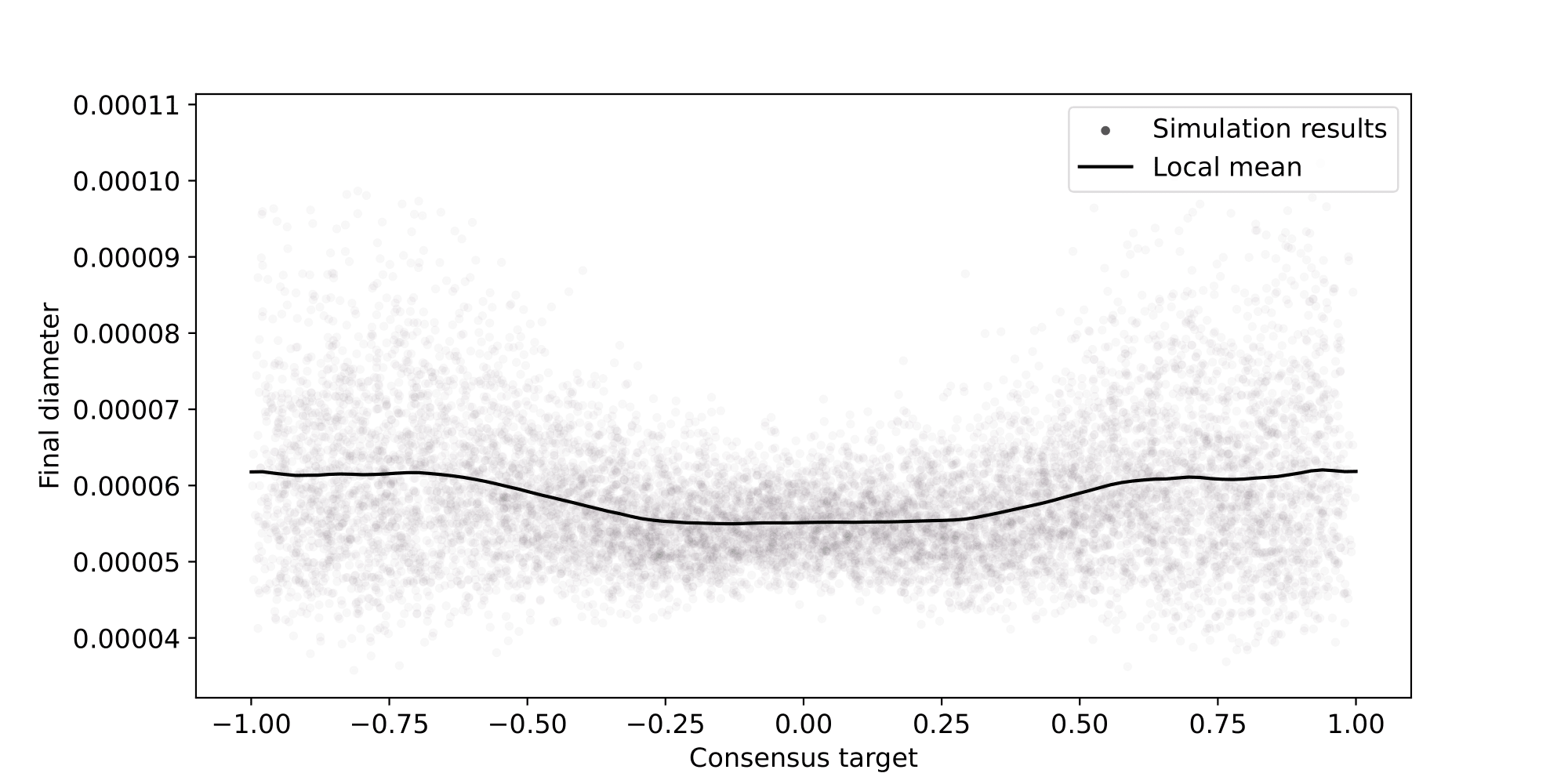}
    \caption{Results of repeated tests of the instantaneous control \eqref{Eqn: instantaneous minimusation of V control strategy} with $\mathcal{S}=1$. Each simulation uses uniformly random initial conditions and a weighted Erdos-Renyi random initial network and runs until opinions have reached a steady state. Each point shows the opinion diameter at the end of the simulation. The black line shows a local mean.}
    \label{fig:final diameter}
\end{figure}

\begin{figure}[ht!]
    \centering
    \includegraphics[width = \linewidth]{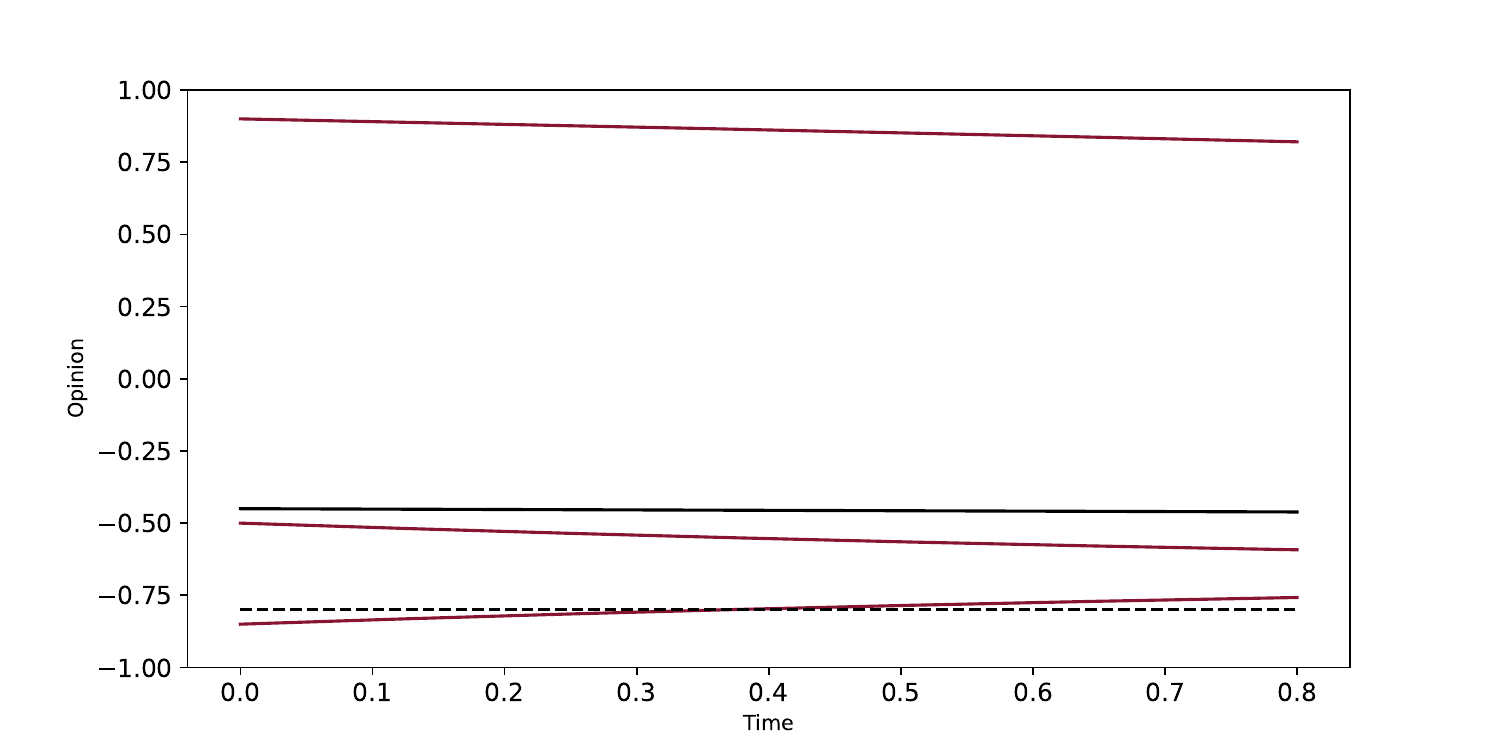}
    \caption{Example trajectory demonstrating that, for certain initial conditions, control to $x^*$ is impossible regardless of the speed of controls ($\mathcal{S}$). Opinion trajectories are shown in red. The target opinion $x^*$ is indicated by a black dashed line. The degree-weighted mean opinion is given by the solid black line and is constant.}
    \label{fig: Instantaneous control abject failure}
\end{figure}

Figure \ref{fig: Instantaneous control abject failure} gives an example trajectory demonstrating that, for certain initial conditions, control to $x^*$ is impossible regardless of the speed of controls ($\mathcal{S}$). A small population size of $N=3$ is chosen for simplicity. Initial opinions are given by $x(0) = (-0.85,-0.5, 0.9)$ with target $x^* = -0.8$ (shown by a dashed black line in Figure \ref{fig: Instantaneous control abject failure}). The initial weights $w(0)$ are chosen to match the steady state of their controlled weight dynamics. That is, for $i\neq j$, $w_{ij}(0)=1$ if $u_{ij}(0) = 1$ and $w_{ij}(0)=0$ if $u_{ij}(0) = -1$, where $u_{ij}$ is determined by \eqref{Eqn: instantaneous minimusation of V control strategy}. The initial network is 
\begin{equation*}
    w(0) = \begin{pmatrix}
        1 & 1 & 0 \\
        1 & 1 & 0 \\
        1 & 1 & 1
    \end{pmatrix} \,.
\end{equation*}

As the degree-weighted mean opinion $\Bar{x}$ (shown by a solid black line in Figure \ref{fig: Instantaneous control abject failure}) does not intersect any individual's opinion trajectory or the target $x^*$, $u_{ij}(t) = u_{ij}(0)$ for all $i,j\in\Lambda$ and $t\in[0,0.8]$. As a result all weights remain constant throughout the dynamics as they begin at the correct steady state. Hence the speed of controls $\mathcal{S}$ is not significant, as the edge weights never change. Note that even though weights begin in a steady state determined by \eqref{Eqn: instantaneous minimusation of V control strategy}, a strategy designed to minimise the distance between $\Bar{x}$ and $x^*$, $\Bar{x}(0)\neq x^*$. This is a result of setting $w_{ii}(t)=1$ for all $i\in\Lambda$ and $t\geq0$, which restricts the possible values of $\Bar{x}$. This represents the idea that each individual always gives their own opinion maximal weight, and that this weight cannot be affected by any control. 

At time $t\approx3.5$ the individual with the lowest opinion crosses $x^*$, meaning that $x^*$ is outside the opinion interval and thus control to consensus at $x^*$ is impossible. Hence for certain initial conditions and consensus targets, control to consensus using the instantaneous control \eqref{Eqn: instantaneous minimusation of V control strategy} is not possible for any value of $\mathcal{S}$. 

\section{Additional Examples of Optimal Control} \label{Appendix: Additional Examples}

In this appendix we show several additional examples of the optimal control problem described in Section \ref{Section: Optimal control}. The figures are presented in the same format as Figure \ref{fig:optimal_control_dynamics_and_control}: the top panel shows the optimal dynamics while the lower panels show several snapshots of the optimal controls. 

Figure \ref{fig:optimal control empty w0} shows the optimal dynamics and controls using the same target and initial opinions as for Figure \ref{fig:optimal_control_dynamics_and_control} but with an empty initial network. As a result the majority of negative controls are replaced with zero control, as there is no need to remove undesired edges. The pattern of positive controls is largely the same as seen in Figure \ref{fig:optimal_control_dynamics_and_control}. 

Figure \ref{fig:optimal control complete w0} shows the optimal dynamics and controls using the same target and initial opinions as for Figure \ref{fig:optimal_control_dynamics_and_control} but with a complete initial network. Here the situation is reversed, as the majority of positive controls are now replaced with zero controls and the pattern of negative controls matches that in Figure \ref{fig:optimal_control_dynamics_and_control}. 

Figure \ref{fig:optimal_control_dynamics WS} shows the optimal dynamics and controls using a different target, different random initial opinions and a new Watts-Strogatz small-world initial network \cite{watts1998collective}. Although the locations of positive and negative controls are different from those in Figure \ref{fig:optimal_control_dynamics_and_control} there is a similar overall pattern: stripes of positive controls whose width matches the interaction function, bringing individuals towards $x^*$; many zero controls, especially towards the end of the dynamics; negative controls that do not form a consistent pattern. 

\begin{figure}[ht!]
    \begin{subfigure}{\linewidth}
        \centering
        \includegraphics[width = .8\linewidth, trim = {1cm 0cm 2cm 1cm}, clip]{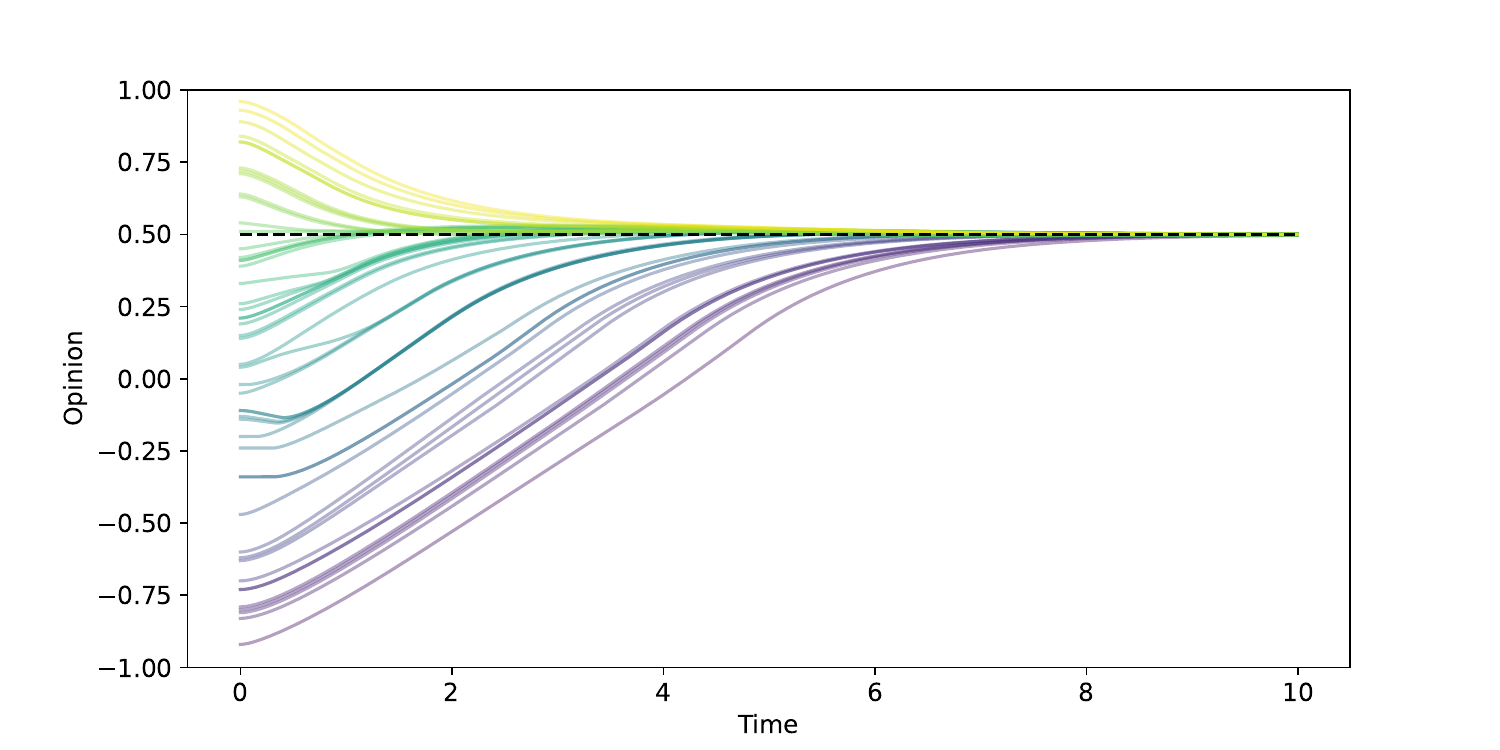}
        \caption{Opinion dynamics under the optimal controls for the cost functional \eqref{Eqn: cost functional}. Opinion trajectories are coloured according to individuals' initial opinions. The target opinion $x^*$ is indicated by a black dashed line. }
        \label{fig:optimal_control_dynamics empty w0}
    \end{subfigure}

    \begin{subfigure}{\linewidth}
        \centering
        \includegraphics[width = 0.9\linewidth, trim = {0.5cm 0.5cm 1.5cm 2.5cm}, clip]{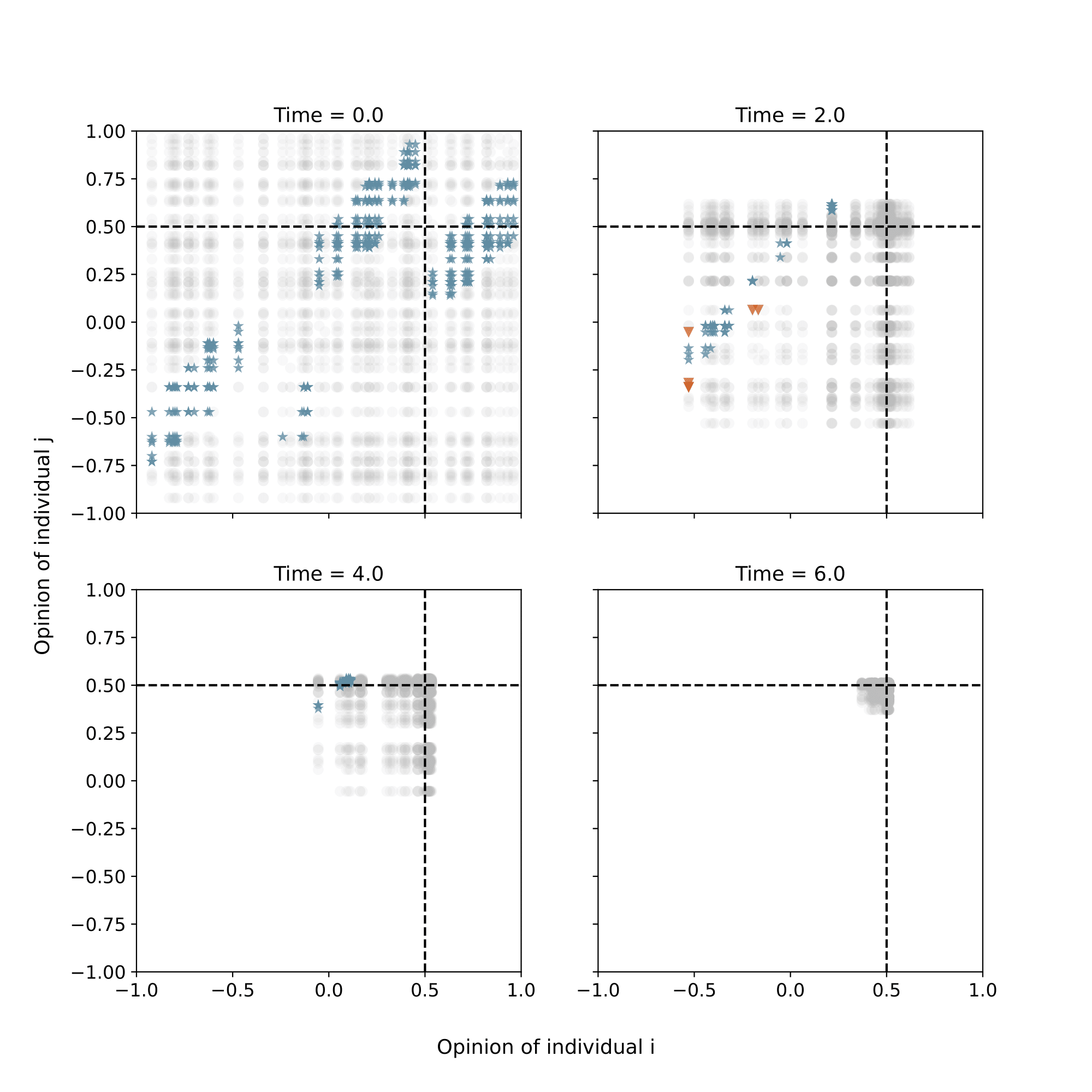}
        \caption{Snapshots of the optimal control at times $t=0,2,4,6$. The horizontal axis gives the opinions $x_i$ for $i\in\Lambda$, the vertical axis gives the opinions $x_j$ for $j\in\Lambda$ and points show the control $u_{ij}$. Blue stars show positive controls, where edges are created/strengthened. Red triangles show negative controls, where edges are being weakened. Grey circles indicate no control. Dashed lines show the location of the target opinion $x^*$, hence as $t$ increases opinions are brought near this value.}
        \label{fig:optimal_control_controls empty w0}
    \end{subfigure}

    \caption{Results of the FBS to find the optimal controls under \eqref{Eqn: cost functional}, using edge weight dynamics of the form \eqref{Eqn: Memory weight controls} with $s$ and $\ell$ given by \eqref{Eqn: ess and ell}. The same initial opinions are used as for the examples in Figure \ref{fig: Controllability from empty network example} and Figure \ref{fig: Controllability from non-empty network example}. The initial network is \textbf{empty}.}
    \label{fig:optimal control empty w0}
\end{figure}

\begin{figure}[ht!]
    \begin{subfigure}{\linewidth}
        \centering
        \includegraphics[width = .8\linewidth, trim = {1cm 0cm 2cm 1cm}, clip]{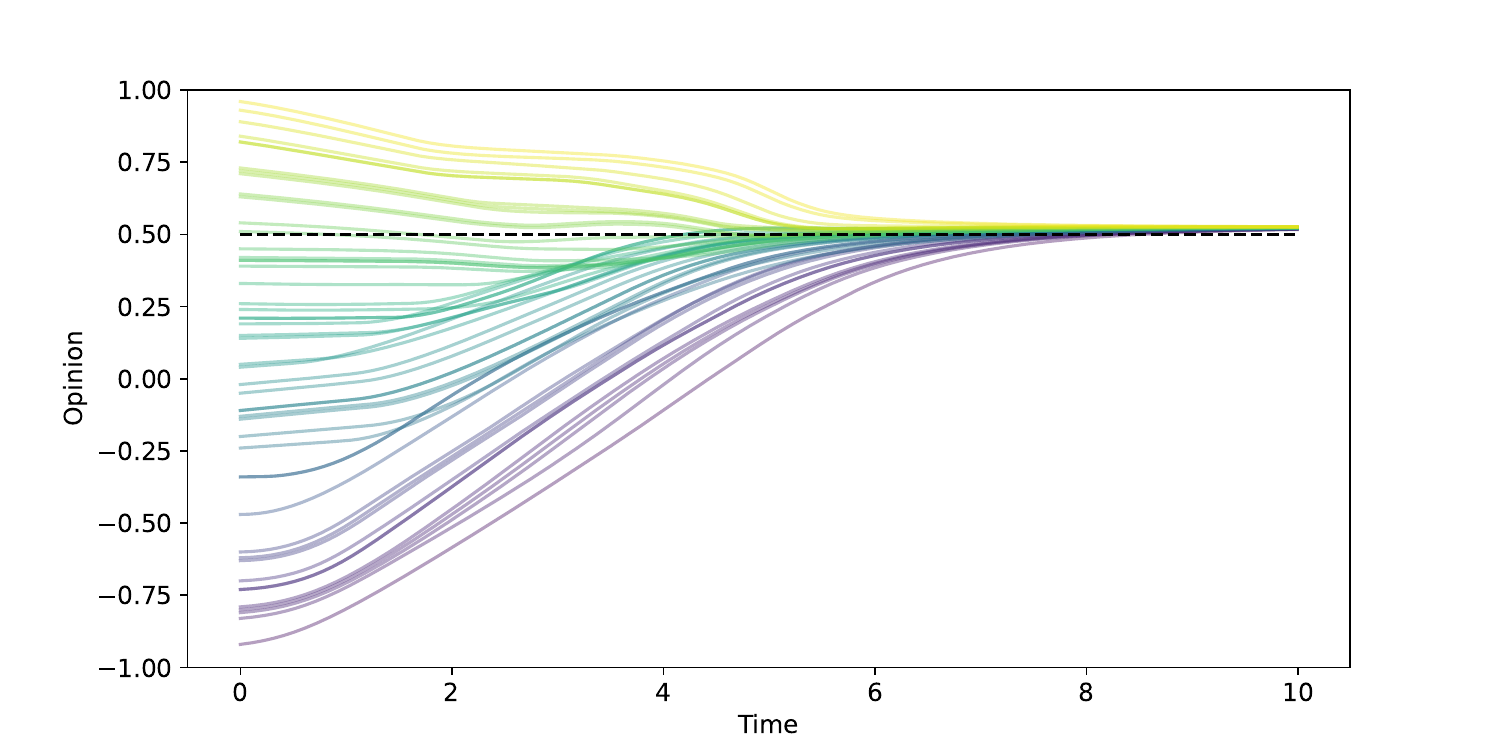}
        \caption{Opinion dynamics under the optimal controls for the cost functional \eqref{Eqn: cost functional}. Opinion trajectories are coloured according to individuals' initial opinions. The target opinion $x^*$ is indicated by a black dashed line. }
        \label{fig:optimal_control_dynamics complete w0}
    \end{subfigure}

    \begin{subfigure}{\linewidth}
        \centering
        \includegraphics[width = 0.9\linewidth, trim = {0.5cm 0.5cm 1.5cm 2.5cm}, clip]{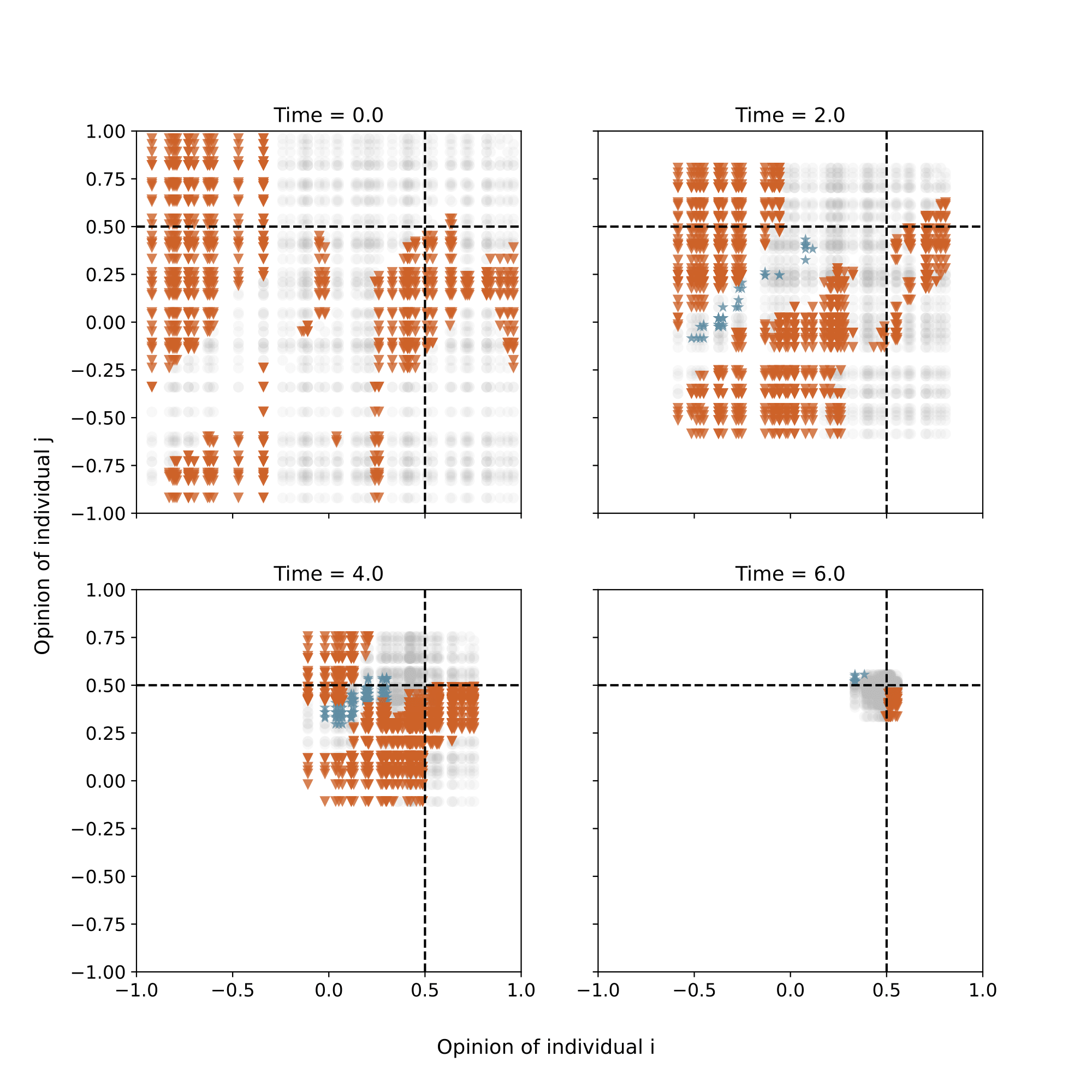}
        \caption{Snapshots of the optimal control at times $t=0,2,4,6$. The horizontal axis gives the opinions $x_i$ for $i\in\Lambda$, the vertical axis gives the opinions $x_j$ for $j\in\Lambda$ and points show the control $u_{ij}$. Blue stars show positive controls, where edges are created/strengthened. Red triangles show negative controls, where edges are being weakened. Grey circles indicate no control. Dashed lines show the location of the target opinion $x^*$, hence as $t$ increases opinions are brought near this value.}
        \label{fig:optimal_control_controls complete w0}
    \end{subfigure}

    \caption{Results of the FBS to find the optimal controls under \eqref{Eqn: cost functional}, using edge weight dynamics of the form \eqref{Eqn: Memory weight controls} with $s$ and $\ell$ given by \eqref{Eqn: ess and ell}. The same initial opinions are used as for the examples in Figure \ref{fig: Controllability from empty network example} and Figure \ref{fig: Controllability from non-empty network example}. The initial network is \textbf{complete}.}
    \label{fig:optimal control complete w0}
\end{figure}

\begin{figure}[ht!]
    \begin{subfigure}{\linewidth}
        \centering
        \includegraphics[width = .8\linewidth, trim = {1cm 0cm 2cm 1cm}, clip]{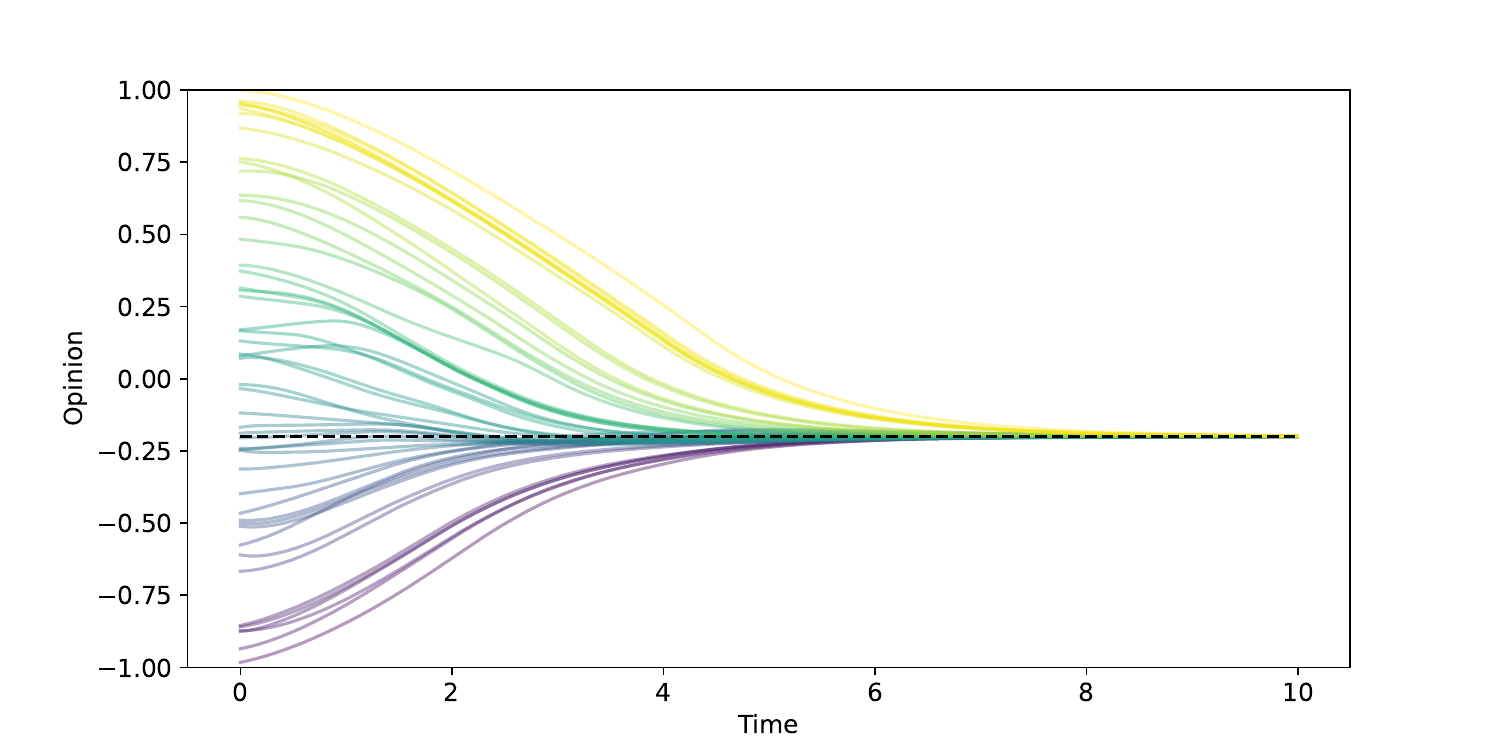}
        \caption{Opinion dynamics under the optimal controls for the cost functional \eqref{Eqn: cost functional}. Opinion trajectories are coloured according to individuals' initial opinions. The target opinion $x^*$ is indicated by a black dashed line. }
        \label{fig:optimal_control_dynamics WS}
    \end{subfigure}

    \begin{subfigure}{\linewidth}
        \centering
        \includegraphics[width = 0.9\linewidth, trim = {0.5cm 0.5cm 1.5cm 2.5cm}, clip]{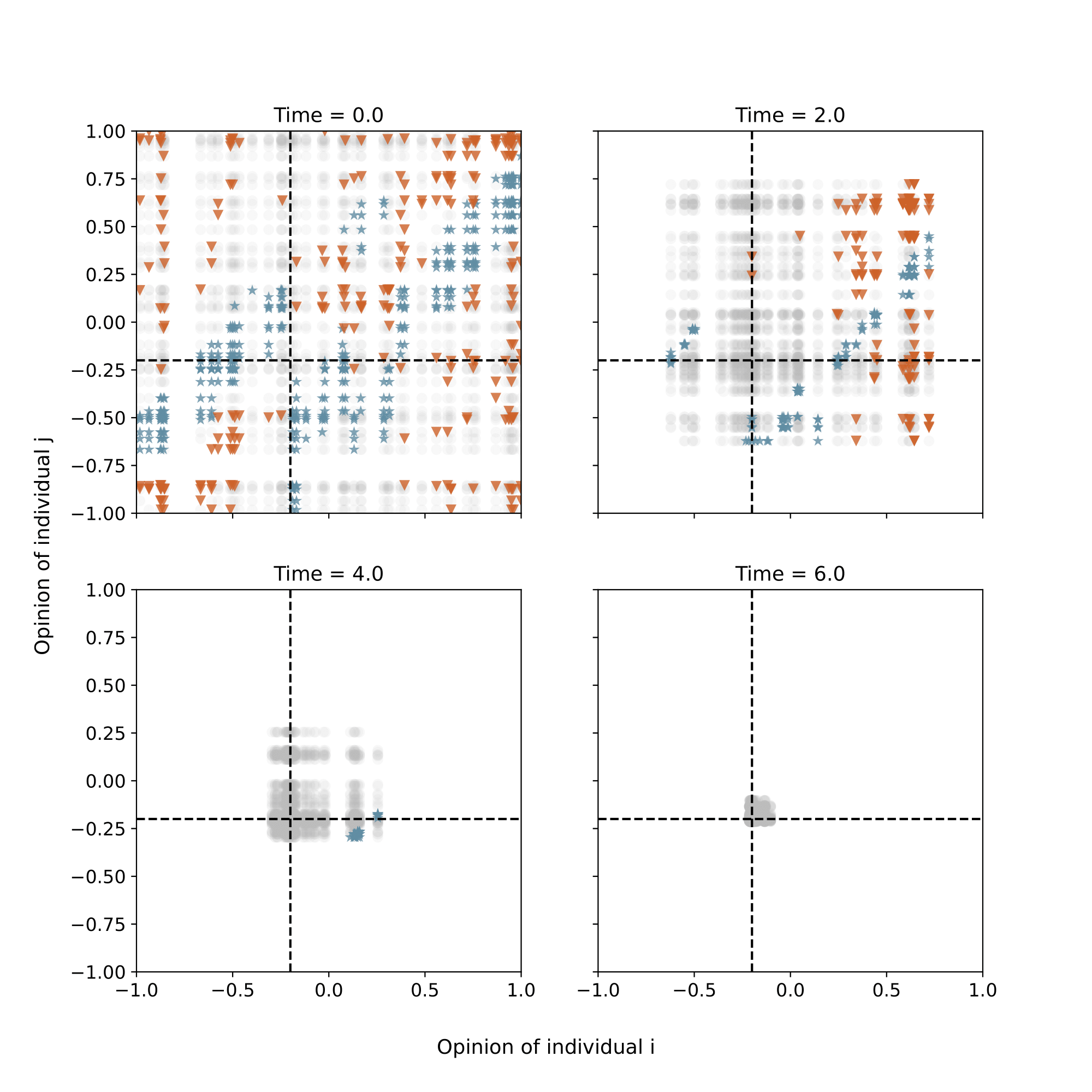}
        \caption{Snapshots of the optimal control at times $t=0,2,4,6$. The horizontal axis gives the opinions $x_i$ for $i\in\Lambda$, the vertical axis gives the opinions $x_j$ for $j\in\Lambda$ and points show the control $u_{ij}$. Blue stars show positive controls, where edges are created/strengthened. Red triangles show negative controls, where edges are being weakened. Grey circles indicate no control. Dashed lines show the location of the target opinion $x^*$, hence as $t$ increases opinions are brought near this value.}
        \label{fig:optimal_control_controls WS}
    \end{subfigure}

    \caption{Results of the FBS to find the optimal controls under \eqref{Eqn: cost functional}, using edge weight dynamics of the form \eqref{Eqn: Memory weight controls} with $s$ and $\ell$ given by \eqref{Eqn: ess and ell}. Initial opinions were sampled uniformly at random n $[-1,1]$. The initial network is a \textbf{Watts-Strogatz random graph}.}
    \label{fig:optimal control WS}
\end{figure}

\end{document}

%% file: preamble.tex
\usepackage[utf8]{inputenc}
\usepackage[dvipsnames]{xcolor}
\usepackage{amsfonts}
\usepackage{amsthm}
\usepackage{amsmath}
\usepackage{parskip}
\usepackage{mathrsfs}
\usepackage{amssymb}
\usepackage{xtab}
\usepackage{setspace}
\usepackage{datetime}
\usepackage{svg}
\usepackage{csquotes}
\usepackage{fancyhdr}
\usepackage{float}
\usepackage{comment}
\usepackage{enumitem}
\usepackage{mathalfa}
\usepackage{array}
\usepackage{titlesec}
\usepackage{mdframed}
\usepackage{listings}
\usepackage{booktabs}
\usepackage{dsfont}
\usepackage{textcomp}
\usepackage{graphicx}
\usepackage{wrapfig}
\usepackage{cases}

\usepackage[font=small,labelfont=bf]{caption}

\usepackage[utf8]{inputenc}
\usepackage[english]{babel}

\usepackage{placeins}
\usepackage{subcaption}
\DeclareMathOperator*{\argmax}{arg\,max}
\DeclareMathOperator*{\argmin}{arg\,min}

\usepackage{hyperref}
\usepackage{cleveref}

\usepackage{cite}

\usepackage{multirow}

\graphicspath{{Figures/}}

\allowdisplaybreaks

\usepackage{thmtools}
\usepackage{thm-restate}
\declaretheorem[name=Proposition,numberwithin=section]{proposition}
\declaretheorem[name=Lemma,numberwithin=section]{lemma}
\declaretheorem[name=Theorem,numberwithin=section]{theorem}
\declaretheorem[name=Definition,numberwithin=section]{definition}
\declaretheorem[name=Remark,numberwithin=section]{remark}

\declaretheorem[name=Assumption]{assumption}

\newcounter{subassumption}

\usepackage{todonotes}
\usepackage{circuitikz}
\usetikzlibrary{decorations.pathreplacing,calligraphy}

%% file: Tikz/tikz_step2.tex
\begin{figure}[H]
\centering
\resizebox{.8\textwidth}{!}{%
\begin{circuitikz}
\tikzstyle{every node}=[font=\Large]
\draw [ color={rgb,255:red,120; green,120; blue,120}, line width=1pt, short] (7.5,11.75) -- (17.5,11.75);
\draw [ color={rgb,255:red,120; green,120; blue,120}, line width=1pt, dashed] (7.5,11.75) -- (6.25,11.75);
\draw [ color={rgb,255:red,120; green,120; blue,120}, line width=1pt, dashed] (17.5,11.75) -- (18.75,11.75);
\draw [line width=1.5pt, short] (12.5,12.25) -- (12.5,11.25);
\draw [ fill={rgb,255:red,0; green,0; blue,0} , line width=0.3pt ] (8,11.75) circle (0.25cm);
\draw [ fill={rgb,255:red,0; green,0; blue,0} , line width=0.6pt ] (10,11.75) circle (0.25cm);
\draw [ fill={rgb,255:red,0; green,0; blue,0} , line width=0.6pt ] (11.5,11.75) circle (0.25cm);
\draw [ fill={rgb,255:red,0; green,0; blue,0} , line width=0.6pt ] (15.5,11.75) circle (0.25cm);
\draw [ fill={rgb,255:red,0; green,0; blue,0} , line width=0.6pt ] (17,11.75) circle (0.25cm);
\draw [line width=1.5pt, ->, >=Stealth] (8,11.75) -- (9,11.75);
\draw [line width=1.5pt, ->, >=Stealth] (10,11.75) -- (10.75,11.75);
\draw [line width=1.5pt, ->, >=Stealth] (17,11.75) -- (16.25,11.75);
\node [font=\Large] at (12.6,12.75) {$x^*$};
\node [font=\Large] at (8,11) {$x_{a-2}$};
\node [font=\Large] at (10,11) {$x_{a-1}$};
\node [font=\Large] at (11.5,12.5) {$x_{a}$};
\node [font=\Large] at (15.5,12.5) {$x_{b}$};
\node [font=\Large] at (17,11) {$x_{b+1}$};

\draw [decorate, decoration = {calligraphic brace,mirror,amplitude=7pt}, line width=1.5pt] (11.5,11.1) --  (15.5,11.1);
\node [font=\large] at (13.5,10.5) {$|x_a - x_b| < r^*$};

\end{circuitikz}
}%

\caption{Diagram for Step 2 of the proof of Proposition \ref{Proposition: Controllability from empty intial network}. In this step individuals are sequentially gathered towards the central individuals $x_a$ and $x_b$ whose opinions are the closest to $x^*$ above and below respectively.}
\label{fig:step 2 diagram}

\end{figure}

%% file: Tikz/tikz_step3.tex
\begin{figure}[H]
\centering
\resizebox{.8\textwidth}{!}{%
\begin{circuitikz}
\tikzstyle{every node}=[font=\Large]
\draw [ color={rgb,255:red,120; green,120; blue,120}, line width=1pt, short] (7.5,11.75) -- (17.5,11.75);
\draw [ color={rgb,255:red,120; green,120; blue,120}, line width=1pt, dashed] (7.5,11.75) -- (6.25,11.75);
\draw [ color={rgb,255:red,120; green,120; blue,120}, line width=1pt, dashed] (17.5,11.75) -- (18.75,11.75);
\draw [line width=1.5pt, short] (12.5,12.25) -- (12.5,11.25);
\draw [ fill={rgb,255:red,0; green,0; blue,0} , line width=0.3pt ] (10.5,11.75) circle (0.25cm);
\draw [ fill={rgb,255:red,0; green,0; blue,0} , line width=0.6pt ] (11,11.75) circle (0.25cm);
\draw [ fill={rgb,255:red,0; green,0; blue,0} , line width=0.6pt ] (11.5,11.75) circle (0.25cm);
\draw [ fill={rgb,255:red,0; green,0; blue,0} , line width=0.6pt ] (15.5,11.75) circle (0.25cm);
\draw [ fill={rgb,255:red,0; green,0; blue,0} , line width=0.6pt ] (16,11.75) circle (0.25cm);
\draw [line width=1.5pt, ->, >=Stealth] (11.5,11.75) -- (12.3,11.75);
\draw [line width=1.5pt, ->, >=Stealth] (15.5,11.75) -- (14.7,11.75);
\node [font=\Large] at (12.6,12.75) {$x^*$};
\node [font=\Large] at (11.5,12.5) {$x_{a}$};
\node [font=\Large] at (15.5,12.5) {$x_{b}$};

\draw [decorate, decoration = {calligraphic brace,mirror,amplitude=7pt}, line width=1.5pt] (10,11.1) --  (16.5,11.1);
\node [font=\large] at (13.5,10.5) {$D(T) < r^*$};

\end{circuitikz}
}%

\caption{Diagram for Step 3 of the proof of Proposition \ref{Proposition: Controllability from empty intial network}. In this step the $N=2$ case is used to control $x_a$ and $x_b$ to $x^*$. All other individuals are connected to one of these two and so follow accordingly.}
\label{fig:step 3 diagram}

\end{figure}